\documentclass[submission,copyright,creativecommons]{eptcs}
\providecommand{\event}{WPTE 2018} % Name of the event you are submitting to
\usepackage{underscore}           % Only needed if you use pdflatex.

\usepackage{multicol}

\DeclareMathAlphabet{\mathcal}{OMS}{cmsy}{m}{n}

\usepackage{theorem}

\newcommand{\eop}{\hspace*{\fill}$\Box$}
\def\qed{\eop}
\newtheorem{definition}{Definition}[section]
\newtheorem{lemma}[definition]{Lemma}
\newtheorem{theorem}[definition]{Theorem}
{\theorembodyfont{\rmfamily}
  \newtheorem{example}[definition]{\it Example}
  \newtheorem{proof}{\it Proof.}
}

\usepackage{amsmath}
\usepackage{amssymb}
\usepackage{stmaryrd}

\def\cR{\mathcal{R}}
\def\cS{\mathcal{S}}
\def\cV{\mathcal{V}}
\newcommand{\funtype}{\Rightarrow}
\def\cI{\mathcal{I}}
\newcommand{\cJ}{\mathcal{J}}
\newcommand{\Interpret}[1]{\llbracket #1 \rrbracket} % need stmaryrd.sty
\def\cC{\mathcal{C}}
\def\cD{\mathcal{D}}
\def\Constrained#1{[\,#1\,]}
\def\Hole{\Box}
\newcommand{\FVar}{{\mathcal{V}\mathit{ar}}}
\newcommand{\Dom}{{\mathcal{D}\mathit{om}}}
\newcommand{\cRcalc}{\cR_{\mathtt{calc}}}
\newcommand{\Sigmaterms}{\Sigma_{\mathit{terms}}}
\newcommand{\Sigmalogic}{\Sigma_{\mathit{theory}}}
\newcommand{\Sigmatheory}{\Sigma_{\mathit{theory}}}
\newcommand{\Sigmacore}{\Sigma_{\mathit{theory}}^{\mathit{core}}}
\newcommand{\Sigmaint}{\Sigma_{\mathit{theory}}^{\INT}}
\def\Val{{\mathcal{V}\mathit{al}}}
\newcommand{\LVar}{{\mathcal{LV}\mathit{ar}}}

\def\sort#1{\mathit{#1}}
\newcommand{\symb}[1]{\mathsf{#1}}
\newcommand{\var}[1]{\mathit{#1}}
\newcommand{\Int}{\mathbb{Z}}
\newcommand{\Bool}{\mathbb{B}}
\newcommand{\INT}{\mathit{int}}
\newcommand{\BOOL}{\mathit{bool}}
\def\Rule#1#2{#1 \to #2}
\def\CRule#1#2#3{#1 \to #2\ \Constrained{#3}}

\newcommand{\SIMPo}{\textsf{SIMP}}
\newcommand{\SIMP}{{\textsf{SIMP}${}^+$}}
\newcommand{\IntVarDecl}[1]{\mathtt{int}~#1}
\newcommand{\IntFunDecl}[2]{\mathtt{int}~#1(#2)}
\newcommand{\lrang}[1]{\mathit{#1}} %< >
\newcommand{\Assign}[2]{#1~=~#2}
\newcommand{\Return}[1]{\mathtt{return}~#1}
\newcommand{\If}[3]{\mathtt{if}(\,#1\,)\{\,#2\,\}\mathtt{else}\{\,#3\,\}}
\newcommand{\Skip}{\epsilon}
\newcommand{\While}[2]{\mathtt{while}(\,#1\,)\{\,#2\,\}}
\newcommand{\True}{\symb{true}} % true
\newcommand{\False}{\symb{false}} % false
\newcommand{\EQ}{\mathrel{\mbox{==}}}
\newcommand{\FCall}[2]{#1(#2)}
\newcommand{\GVar}{{\mathcal{GV}\mathit{ar}}}

\newcommand{\toExpr}{\Downarrow_{\mathrm{calc}}}
\newcommand{\Config}[3]{\langle #1,\, #2,\, #3 \rangle}
\newcommand{\Update}[3]{#1[ #2 \mapsto #3 ]}

\newcommand{\Inference}[2]{\displaystyle \frac{~~ #1 ~~}{~~ #2 ~~}}

\newcommand{\Convert}{\mathit{conv}}
\newcommand{\Trans}[1][P]{\mathit{aux}_{#1}}
\def\Xvec{\vec{x}}
\newcommand{\Yvec}[1][]{\overrightarrow{y_{#1}}}
\newcommand{\Zvec}[1][]{\overrightarrow{z_{#1}}}
\newcommand{\Evec}[1][]{\overrightarrow{e'_{#1}}}
\def\PrgrmA{P_1}
\def\cRsumiter{\cR_1}
\def\cRseq{\cR_2}
\def\cRsumrec{\cR_3}
\def\cRsumrecstack{\cR_4}
\def\SumOneA{\symb{u_1}}
\def\Ga{\symb{u_2}}
\def\Gb{\symb{u_3}}
\def\Gc{\symb{u_4}}
\def\Gnum{\symb{u_2'}}
\def\SumA{\symb{u_1}}
\def\SumB{\symb{u_2}}
\def\SumC{\symb{u_3}}
\def\SumD{\symb{u_4}}
\def\SumE{\symb{u_5}}
\def\SumF{\symb{u_6}}
\def\SumG{\symb{u_7}}
\def\SumH{\symb{u_8}}
\def\SumI{\symb{u_9}}
\def\MainA{\symb{u_{10}}}
\def\MainB{\symb{u_{11}}}
\def\MainC{\symb{u_{12}}}

\title{On Transforming Functions Accessing Global Variables into Logically Constrained Term Rewriting Systems%
\thanks{%
This work was partially supported by DENSO Corporation, NSITEXE, Inc., and JSPS KAKENHI Grant Number JP18K11160.}
}
\author{Yoshiaki Kanazawa 
\institute{Graduate School of Informatics\\
Nagoya University\\
Nagoya, Japan}
\email{yoshiaki@trs.css.i.nagoya-u.ac.jp}
\and
Naoki Nishida
\institute{Graduate School of Informatics\\
Nagoya University\\
Nagoya, Japan}
\email{nishida@i.nagoya-u.ac.jp}
}
\def\titlerunning{On Transforming Functions Accessing Global Variables into LCTRSs}
\def\authorrunning{Y. Kanazawa and N. Nishida}

\begin{document}
\maketitle

\begin{abstract}
In this paper, we show a new approach to transformations of an imperative program with function calls and global variables into a logically constrained term rewriting system. The resulting system represents transitions of the whole execution environment with a call stack. More precisely, we prepare a function symbol for the whole environment, which stores values for global variables and a call stack as its arguments. For a function call, we prepare rewrite rules to push the frame to the stack and to pop it after the execution. Any running frame is located at the top of the stack, and statements accessing global variables are represented by rewrite rules for the environment symbol. We show a precise transformation based on the approach and prove its correctness.
\end{abstract}

%%%%%%%%%%%%%%%%%%%%%%%%%%%%%%%%%%%%%%%%%%%%%%%%%%%%%%%%%%%%%%%%%%%%%%%%%
\section{Introduction}
\label{sec:intro}
%%%%%%%%%%%%%%%%%%%%%%%%%%%%%%%%%%%%%%%%%%%%%%%%%%%%%%%%%%%%%%%%%%%%%%%%%

Recently, analyses of imperative programs (written in C, Java Bytecode, etc.) via transformations into term rewriting systems have been investigated \cite{FK09,FKS11,FNSKS08b,Otto10}.
In particular, \emph{constrained rewriting systems} are popular for these transformations, since logical constraints used for modeling the control flow can be separated from terms expressing intermediate states~\cite{FK09,FKS11,FNSKS08b,NNKSS11,SNSSK09}.
To capture the existing approaches for constrained rewriting in one setting, the framework of a \emph{logically constrained term rewriting system} (an LCTRS, for short) has been proposed~\cite{KN13frocos}. 
Transformations of C programs with integers, characters, arrays of integers, global variables, and so on into LCTRSs have been discussed in~\cite{FKN17tocl}.

%Transformations of simple imperative programs, e.g., C programs over the integers, into constrained rewriting systems have been developed, and been extended to many practical fragments, e.g., function calls, arrays of integers, and global variables.
A basic idea of transforming functions defined in simple imperative programs over the integers, so-called \emph{while} programs, is to represent transitions of parameters and local variables as rewrite rules with auxiliary function symbols.
The resulting rewriting system can be considered a \emph{transition system} w.r.t.\ parameters and local variables.
Consider the function $\mathtt{sum1}$ in Figure~\ref{fig:sum1}, which is written in the C language.
The function \texttt{sum1} computes the summation from $0$ to a given non-negative integer $x$.
The execution of the body of this function can be considered a transition of values for $x$, $i$, and $z$, respectively.
For example, we have the following transition for \verb+sum1(3)+:
\[
	(3,0,0)
	\mathrel{\to}
	(3,0,1)
	\mathrel{\to}
	(3,1,1)
	\mathrel{\to}
	(3,1,3)
	\mathrel{\to}
	(3,2,3)
	\mathrel{\to}
	(3,2,6)
	\mathrel{\to}
	(3,3,6)
	\mathrel{\to}
    (3,3,6)
\]
This transition for the execution of the function $\mathtt{sum1}$ can be modeled by an LCTRS as follows~\cite{FNSKS08b,FKN17tocl}:
\[
	\cRsumiter =
 \left\{
 \begin{array}{r@{\>}c@{\>}lr@{\,}c@{\,}l}
 	\Rule{\symb{sum1}(x) &}{& \SumOneA(x,\symb{0},\symb{0})}, \\
 	\CRule{\SumOneA(x,i,z) &}{& \SumOneA(x,i+\symb{1},z+i+\symb{1}) &}{& i < x &}, \\
 	\CRule{\SumOneA(x,i,z) &}{& \symb{return}(z) &}{& \neg (i < x) &} \\
 \end{array}
 \right\}
\]
Note that the auxiliary function symbol $\SumOneA$ can be considered locations stored in the \emph{program counter}.
The transformed LCTRS is useful to verify the original program~\cite{FKN17tocl}.
For example, the theorem proving method based on \emph{rewriting induction}~\cite{Red90} can automatically prove that $\forall n \in \Int.\ \symb{sum1}(n) = \frac{n(n+1)}{2}$, i.e., correctness of the C program~\cite{SNSSK09,KN15lpar,FKN17tocl}.

\begin{figure}[tb]
\begin{verbatim}
  int sum1(int x){
    int i = 0;  
    int z = 0;
    for( i = 0 ; i < x ; i = i + 1 ){
      z = z + i + 1;
    }
    return z;
  }
\end{verbatim}
%\vspace{-10pt}	
\caption{a C program defining a function to compute the summation from $0$ to $x$.}
\label{fig:sum1}
\end{figure}

A function call is added as an extra argument of the auxiliary symbol that corresponds to the statement of the call.
Let us consider the following function in addition to $\mathtt{sum1}$ in Figure~\ref{fig:sum1}:
\begin{verbatim}
    int g(int x){
      int z = 0;
      z = sum1(x);
      return x * z;
    }
\end{verbatim}
This function is transformed into the following rules:
\[
% \cRg =
 \{~
 	\Rule{\symb{g}(x)}{\Ga(x,\symb{0})},
~~
 	\Rule{\Ga(x,z)}{\Gb(x,z,\symb{sum1}(x))},
~~
 	\Rule{\Gb(x,z,\symb{return}(y))}{\Gc(x,y)},
~~
 	\Rule{\Gc(x,z)}{\symb{return}(x \times z)}
 ~\}
\]
The auxiliary function symbol $\Ga$ calls $\symb{sum1}$ in the third argument of $\Gb$ by means of the rule for $\Ga$.

To deal with a global variable under sequential execution, it is enough to pass a value stored in the global variable to a function call as an extra argument and to receive from the called function a value of the global variable that may be updated in executing the function call, restoring the value in the global variable.
Let us add a global variable counting the total number of function calls to the above program as in Figure~\ref{fig:sum1-g}.
%When we consider the sequential execution, 
This program is transformed into the following LCTRS~\cite{FKN17tocl}:
\[
	\cRseq =
 \left\{
 \begin{array}{r@{\>}c@{\>}lr@{\,}c@{\,}l}
 	\Rule{\symb{sum1}(x,\var{num}) &}{& \SumOneA(x,\symb{0},\symb{0},\var{num}+\symb{1})}, \\
 	\CRule{\SumOneA(x,i,z,\var{num}) &}{& \SumOneA(x,i+\symb{1},z+i+\symb{1},\var{num}) &}{& i < x &}, \\
 	\CRule{\SumOneA(x,i,z,\var{num}) &}{& \symb{return}(z,\var{num}) &}{& \neg (i < x) &}, \\[5pt]
 	\Rule{\symb{g}(x,\var{num}) &}{& \Ga(x,\var{num},\symb{0})}, \\
 	\Rule{\Ga(x,\var{num},z) &}{& \Gnum(x,\var{num}+\symb{1},z)}, \\
 	\Rule{\Gnum(x,\var{num},z) &}{& \Gb(x,\var{num},z,\symb{sum1}(x,\var{num}))}, \\
 	\Rule{\Gb(x,\var{num_{old}},z,\symb{return}(y,\var{num_{new}})) &}{& \Gc(x,\var{num_{new}},y)},\\
 	\Rule{\Gc(x,\var{num},z) &}{& \symb{return}(x \times z,\var{num})} \\
 \end{array}
 \right\}
\]
The above approach to transformations of function calls is very naive but not general.
For example, to model parallel execution, a value stored in a global variable does not have to be passed to a particular function or a process because another function or process may access the global variable.

\begin{figure}[t]
%\begin{verbatim}
%  int num = 0;
%
%  int sum1(int x){
%    num = num + 1;
%    int i = 0;
%    int z = 0;
%    for( i = 0 ; i < x ; i = i + 1 ){
%      z += i + 1;
%    }
%    return z;
%  }
%
%  int g(int x){
%    num = num + 1;
%    return x * sum1(x);
%  }
%\end{verbatim}
\begin{multicols}{2}
\begin{verbatim}
  int num = 0;

  int sum1(int x){
    num = num + 1;
    int i = 0;
    int z = 0;
    for( i = 0 ; i < x ; i = i + 1 ){
      z = z + i + 1;
    }
    return z;
  }
\end{verbatim}
\columnbreak
\begin{verbatim}    
  int g(int x){
    int z = 0;
    num = num + 1;
    z = sum1(x);
    return x * z;
  }
\end{verbatim}
\end{multicols}
\vspace{-5pt}
\caption{a C program obtained by adding the definition of \texttt{g} into the program for \texttt{sum}.}
\label{fig:sum1-g}
\end{figure}

In this paper, %using an example, 
we show another approach to transformations of imperative programs with function calls and global variables into LCTRSs.
% in consideration of parallel execution. In particular, aim of this transformation is modeling parallel processes with shared resources such as global variables using LCTRSs. 
%We also succeeded in proof of termination of LCTRS obtained using a new transformation method by the tool developed by our research group this fiscal year.
Our target languages are call-by-value imperative languages such as C. %, and our transformation has a close connection with call-by-value evaluation.
For this reason, we use a small subclass of C programs over the integers as fundamental imperative programs.
We show a precise transformation along the approach and prove its correctness.

Our idea of the treatment for global variables in calling functions is to prepare a new symbol to represent the whole environment for execution.
Values of global variables are stored in arguments of the new symbol, and transitions accessing global variables are represented as transitions of the environment.
In reduction sequences of LCTRSs obtained by the original transformation, positions of function calls are not unique, and thus, we may need (possibly infinitely) many rules for a transition related to a global variable.
To solve this problem, we prepare a so-called \emph{call stack}, and transform programs into LCTRSs that specify statements as rewrite rules for not only user-defined functions but also the introduced symbol of the environment.
In calling a function, a frame of the called function is pushed to the stack, and popped from the stack when the execution halts successfully.
This implies that any running frame is located at the top of the stack, i.e., positions of function calls are unique.
We transform statements not accessing global variables into rewrite rules for called functions as well as the previous transformation, and transform statements accessing global variables into rewrite rules for the introduced symbol for the environment. 
%The introduction of the stack may make it difficult to prove termination of the transformed LCTRS compared with the original transformation. 

This paper is organized as follows.
In Section~\ref{sec:preliminaries}, we recall LCTRSs and a small imperative language.
In Section~\ref{sec:new-transformation}, using an example, we show a new approach to transformations of imperative programs into LCTRSs.
In Section~\ref{sec:formalization}, we precisely define a transformation and show its correctness. 
In Section~\ref{sec:conclusion}, we %summarize this paper and 
describe a future direction of this research.

%%%%%%%%%%%%%%%%%%%%%%%%%%%%%%%%%%%%%%%%%%%%%%%%%%%%%%%%%%%%%%%%%%%%%%%%%
\section{Preliminaries}
\label{sec:preliminaries}
%%%%%%%%%%%%%%%%%%%%%%%%%%%%%%%%%%%%%%%%%%%%%%%%%%%%%%%%%%%%%%%%%%%%%%%%%

In this section, we recall 
%\emph{logically constrained term rewriting systems} (LCTRS, for short)
LCTRSs, following the definitions in~\cite{KN13frocos,FKN17tocl}.
We also recall a small imperative language {\SIMP} with global variables and function calls. %, and then introduce a conversion of \emph{while} programs to LCTRSs. 
Familiarity with basic notions on term rewriting~\cite{BN98,Ohl02} is assumed.

\subsection{Logically Constrained Term Rewriting Systems}

%\paragraph{Many-sorted terms}
%We introduce terms, typing, substitutions, contexts, and subterms (with corresponding terminology) in the usual way for many-sorted term rewriting.

Let $\cS$ be a set of \emph{sorts} and $\cV$ a countably infinite set of \emph{variables}, each of which is equipped with a sort.
A \emph{signature} $\Sigma$ is a set, disjoint from $\cV$, of \emph{function symbols} $f$, each of which is equipped with a \emph{sort declaration} $\iota_1 \times \cdots \times \iota_n \funtype \iota$ where $\iota_1,\ldots,\iota_n,\iota \in \cS$.
For readability, we often write $\iota$ instead of
$\iota_1 \times \cdots \times \iota_n  \funtype \iota$ if $n=0$.
We denote the set of well-sorted \emph{terms}  %\Dropped{$s$}
% $s$, written as $\vdash s : \iota$, 
over $\Sigma$ and $\cV$
by $T(\Sigma,\cV)$.
% contains any expression $s$ such that $\vdash s : \iota$ can be derived for some sort $\iota$, using the following inference rules:
%\[
%\frac{~ x : \iota \in \cV ~}{~ \vdash x : \iota ~}
%\quad
%\frac{~ f : \iota_1 \times \cdots \times \iota_n \funtype \iota
%\in \Sigma \quad \vdash s_1 : \iota_1 \quad \ldots \quad \vdash s_n : \iota_n ~}{~ \vdash f(s_1,\ldots,s_n) : \iota ~}
%\]
In the rest of this section, we fix $\cS$, $\Sigma$, and $\cV$.
%Note that for every term $s$, there is a unique sort $\iota$ with $\vdash s : \iota$.
%Let $\vdash s : \iota$.
%We call $\iota$ the \emph{sort of} $s$.
The set of variables occurring in $s_1,\ldots,s_n$ is denoted by $\FVar(s_1,\ldots,s_n)$.
% and $s$ is said to be \emph{ground} if $\FVar(s) = \emptyset$.
Given a term $s$ and a \emph{position} $p$ (a sequence of positive integers) of $s$, $s|_p$ denotes the subterm of $s$ at position $p$,
and
%If $\vdash s|_p : \iota$ and $\vdash t : \iota$, then 
$s[t]_p$ denotes $s$ with the subterm at position $p$ replaced by $t$.
A \emph{context} $C[~]$ is a term containing one \emph{hole} $\Hole_\iota : \iota$.
For a term $s : \iota$, $C[s]$ denotes the term obtained from $C[\,]$ by replacing $\Hole_\iota$ by $s$.

A \emph{substitution} $\gamma$ is a sort-preserving total mapping from $\cV$ to $T(\Sigma,\cV)$, and naturally extended for a mapping from $T(\Sigma,\cV)$ to $T(\Sigma,\cV)$:
the result $s\gamma$ of applying a substitution $\gamma$ to a term $s$ is $s$ with all occurrences of a variable $x$ replaced by $\gamma(x)$.
The \emph{domain} $\Dom(\gamma)$ of $\gamma$ is the set of variables $x$ with $\gamma(x) \neq x$.
The notation $\{ x_1 \mapsto s_1, \ldots, x_k \mapsto s_k \}$ denotes a substitution $\gamma$ with $\gamma(x_i) = s_i$ for $1 \leq i \leq n$, and $\gamma(y) = y$ for $y \notin \{x_1, \ldots,x_n\}$.
%For two substitutions $\gamma$ and $\delta$, their composition $\gamma\delta$ is given by $x(\gamma\delta) = \delta(\gamma(x))$ for all variables $x$.
%
%Two terms $s$ and $t$ are \emph{unifiable} if there exists a substitution $\gamma$ such that $s\gamma = t\gamma$.
%Then, $\gamma$ is called a \emph{unifier} for $s$ and $t$.
%We call $\gamma$ a \emph{most general unifier} (mgu) for $s$ and $t$ if for any unifier $\sigma$ for $s$ and $t$ there is a substitution $\delta$ such that $\sigma = \delta \gamma$.

%Given a term $s$, a \emph{position} in $s$ is a sequence $p$ of positive integers such that $s|_p$ is defined, where $s|_\varepsilon = s$ and $f(s_1,\ldots,s_n)|_{i\cdot p} = (s_i)|_p$. 
%We call $s|_p$ a \emph{subterm} of $s$.

%\paragraph{Logical terms}
To define LCTRSs, we consider different kinds of symbols and terms:
%We assume given:
%\begin{itemize}
%\item 
(1) two signatures $\Sigmaterms$ and $\Sigmalogic$ such that $\Sigma = \Sigmaterms \cup \Sigmalogic$,
%\item 
(2) a mapping $\cI$ which assigns to each sort $\iota$ occurring in $\Sigmatheory$ a set $\cI_\iota$,
%\item 
(3) a mapping $\cJ$ which assigns to each $f : \iota_1 \times \cdots \times \iota_n \funtype \iota \in \Sigmalogic$ a function in $\cI_{\iota_1} \times \cdots \times \cI_{\iota_n} \funtype \cI_\iota$,
and
%\item 
(4) a set $\Val_\iota \subseteq \Sigmalogic$ of \emph{values}---function symbols $a : \iota$ such that $\cJ$ gives a bijective mapping from $\Val_\iota$ to $\cI_\iota$---for each sort $\iota$ occurring in $\Sigmatheory$.
%\end{itemize}
We require that $\Sigmaterms \cap \Sigmalogic \subseteq \Val = \bigcup_{\iota \in \cS} \Val_\iota$.
The sorts occurring in $\Sigmalogic$ are called \emph{theory sorts}, and the symbols \emph{theory symbols}.
Symbols in ${\Sigmalogic} \setminus {\Val}$ are \emph{calculation symbols}.
A term in $T(\Sigmalogic,\cV)$ is called a \emph{theory term}.
For ground theory terms, we define the interpretation as $\Interpret{f(s_1,\ldots,s_n)} = \cJ(f)(\Interpret{s_1},\ldots,\Interpret{s_n})$.
For every ground theory term $s$, there is a unique value $c$ such that $\Interpret{s} = \Interpret{c}$. %\Dropped{We say that $c$ is the value of $s$.}
%\Dropped{Terms in $T(\Sigmaterms,\emptyset)$ can be thought of as the primary objects of rewriting:a reduction typically begins and ends with such terms, with elements of ${\Sigmalogic} \setminus {\Val}$ to perform calculations in the underlying theory.}
We use infix notation for theory and calculation symbols.

A \emph{constraint} is a theory term $\varphi$ of some sort $\BOOL$ with $\cI_\BOOL = \Bool = \{ \top,\bot \}$, the set of \emph{booleans}.
A constraint $\varphi$ is \emph{valid} if $\Interpret{\varphi\gamma} = \top$ for all substitutions $\gamma$ which map $\FVar(\varphi)$ to values, and \emph{satisfiable} if $\Interpret{\varphi\gamma} = \top$ for some such substitution.
A substitution $\gamma$ \emph{respects} $\varphi$ if $\gamma(x)$ is a value for all $x \in \FVar(\varphi)$ and $\Interpret{\varphi\gamma} = \top$.
We typically choose a theory signature with $\Sigmalogic \supseteq \Sigmacore$, where $\Sigmacore$ contains
$\symb{true},\symb{false} : \BOOL$, $\wedge, \vee, \implies : \BOOL \times \BOOL \funtype \BOOL$, $\neg: \BOOL \funtype \BOOL$, and, for all theory sorts $\iota$, symbols $=_\iota, \neq_\iota : \iota \times \iota \funtype \BOOL$, and an evaluation function $\cJ$ that interprets these symbols as expected. 
We omit the sort subscripts from $=$ and $\neq$ when they are clear from context.

The standard integer signature $\Sigmaint$ is $\Sigmacore \cup \{ +, -,*,\symb{exp},\symb{div}, \symb{mod} : \INT \times \INT \funtype \INT \}\cup \{ {\geq}, {>} : \INT \times \INT \funtype \BOOL \} \cup \{ \symb{n} : \INT \mid n \in \Int \}$ with values $\symb{true}$, $\symb{false}$, and $\symb{n}$ for all integers $n \in \Int$.
Thus, we use $\symb{n}$ (in \textsf{sans-serif} font) as the function symbol for $n \in \Int$ (in $\mathit{math}$ font).
We define $\cJ$ in the natural way, except: since all $\cJ(f)$ must be total functions, we set $\cJ(\symb{div})(n,0) = \cJ(\symb{mod})(n,0) = \cJ(\symb{exp})(n,k) = 0$ for all $n$ and all $k < 0$. 
%Of course, when constructing LCTRSs, we normally add explicit error checks to prevent such calls.
When constructing LCTRSs from, e.g., \emph{while} programs, we can add explicit error checks for, e.g., ``division by zero'', to constraints (cf.~\cite{FKN17tocl}).

%\paragraph{Rules and rewriting}

%We adapt the standard notions of rewriting (see, e.g., \cite{BN98}) by including constraints and adding rules to perform calculations.

A \emph{constrained rewrite rule} is a triple $\CRule{\ell}{r}{\varphi}$ such that $\ell$ and $r$ are terms of the same sort, $\varphi$ is a constraint, and $\ell$ has the form $f(\ell_1,\dots,\ell_n)$ and contains at least one symbol in $\Sigmaterms \setminus \Sigmatheory$ (i.e., $\ell$ is not a theory term).
If $\varphi = \symb{true}$ with $\cJ(\symb{true}) = \top$, we may write $\Rule{\ell}{r}$.
We define $\LVar(\CRule{\ell}{r}{\varphi})$ as $\FVar(\varphi) \cup (\FVar(r) \setminus \FVar(\ell))$.
We say that a substitution $\gamma$ \emph{respects} $\CRule{\ell}{r}{\varphi}$ if $\gamma(x) \in \Val$ for all $x \in \LVar(\CRule{\ell}{r}{\varphi})$, and $\Interpret{\varphi\gamma} = \top$.
%The rule is \emph{left-linear} if $\ell$ is linear, i.e., all variables occur at most once in $\ell$, and \emph{irregular} if $\FVar(\varphi) \setminus \FVar(\ell) \neq \emptyset$.
Note that it is allowed to have $\FVar(r) \not\subseteq \FVar(\ell)$, but fresh variables in the right-hand side may only be instantiated with \emph{values}.
%This is done to model user input or random choice.
%Otherwise, variables outside the constraint may be instantiated by any term.
%we do not impose strategies like innermost or call-by-value reduction.
%
Given a set $\cR$ of constrained rewrite rules, we let $\cRcalc$ be the set $\{ \CRule{f(x_1,\ldots,x_n)}{y}{y = f(x_1,\ldots,x_n)} \mid f : \iota_1 \times \cdots \times \iota_n \funtype \iota \in {\Sigmalogic} \setminus {\Val} \}$.
We usually call the elements of $\cRcalc$ constrained rewrite rules (or \emph{calculation rules}) even though their left-hand side is a theory term.
The \emph{rewrite relation} $\to_{\cR}$ is a binary relation on terms, defined by:
%\begin{quotation}
$s[\ell\gamma]_p \mathrel{\to_\cR} s[r\gamma]_p$ if 
$\CRule{\ell}{r}{\varphi} \in \cR \cup \cRcalc$ and $\gamma$ respects $\CRule{\ell}{r}{\varphi}$.
We may say that the reduction occurs at position $p$.
%\end{quotation}
%where $C[~]$ is a context with exactly one hole.
%We say that the reduction occurs at position $p$ if $C[~] = C[\Box]_p$. 
%Let $s \mathrel{\leftrightarrow_\cR} t$ if $s \mathrel{\to_\cR} t$ or $t \mathrel{\to_\cR} s$.
A reduction step with $\cRcalc$ is called a \emph{calculation}.
%A term is in \emph{normal form} if it cannot be reduced with $\to_\cR$.
%We say that $t$ is a \emph{normal form of $s$} if $s \mathrel{\to_\cR} t$ and $t$ is a normal form.
%The relation $\to_\cR$ is \emph{confluent} if whenever $s \mathrel{\to_\cR^*} t_1$ and $s \mathrel{\to_\cR^*} t_2$, there exists also some $u$ with $t_1 \mathrel{\to_\cR^*} u$ and $t_2 \mathrel{\to_\cR^*} u$.
%Note that if $\to_\cR$ is confluent, every term has at most one normal form (intuitively, then $\cR$ is deterministic with respect to big-step semantics).

Now we define a \emph{logically constrained term rewriting system} (an LCTRS, for short) as the abstract rewriting system $(T(\Sigma,\cV),\to_\cR)$ which is simply written by $\cR$.
An LCTRS is usually given by supplying $\Sigma$, $\cR$, and an informal description of $\cI$ and $\cJ$ if these are not clear from context.
An LCTRS $\cR$ is said to be \emph{left-linear} if for every rule in $\cR$, the left-hand side is linear.
$\cR$ is said to be \emph{non-overlapping} if
  for every term $s$ and rule $\CRule{\ell}{r}{\varphi}$ such that $s$ reduces with $\CRule{\ell}{r}{\varphi}$ at the root
  position: (a) there are no other rules $\CRule{\ell'}{r'}{\varphi'}$ such that $s$ reduces
  with $\CRule{\ell'}{r'}{\varphi'}$ at the root position, and (b) if $s$ reduces with any
  rule at a non-root position $q$, then $q$ is not a position of
  $\ell$.
$\cR$ is said to be \emph{orthogonal} if $\cR$ is left-linear and non-overlapping.
For $\CRule{f(\ell_1,\ldots,\ell_n)}{r}{\varphi} \in \cR$, we call $f$ a \emph{defined symbol} of $\cR$, 
and non-defined elements of $\Sigmaterms$ and all values are called \emph{constructors} of $\cR$.
Let $\cD_\cR$ be the set of all defined symbols and $\cC_\cR$ the set of constructors.
A term in $T(\cC_\cR,\cV)$ is a \emph{constructor term} of $\cR$.
We call $\cR$ a \emph{constructor system} if the left-hand side of each rule $\CRule{\ell}{r}{\varphi} \in \cR$ is of the form $f(t_1,\ldots,t_n)$ with $t_1,\ldots,t_n$ constructor terms.

\begin{example}[\cite{FKN17tocl}]
\label{ex:factsignature}
Let $\cS = \{ \INT,\BOOL \}$, and
$\Sigma = \Sigmaterms \cup \Sigmaint$, where
$
\Sigmaterms = \{~ \symb{fact} : \INT \funtype \INT ~\} \cup \{~ \symb{n} : \INT \mid n \in \Int ~\}
$.
Then both $\INT$ and $\BOOL$ are theory sorts.
We also define set and function interpretations, i.e., $\cI_\INT = \Int$, $\cI_\BOOL = \Bool$, and $\cJ$ is defined as above. 
%With $=$ for $=_\INT$ and infix notation, 
Examples of theory terms are $\symb{0} = \symb{0}+\symb{-1}$ and
$x+\symb{3} \geq y + -\symb{42}$ that are constraints.
$\symb{5}+\symb{9}$ is also a (ground) theory term, but not a constraint.
Using calculation steps, a term $\symb{3}-\symb{1}$ reduces to $\symb{2}$ in one step with the calculation rule $\CRule{x-y}{z}{z = x-y}$, and $\symb{3} \times (\symb{2} \times (\symb{1} \times \symb{1}))$ reduces to $\symb{6}$ in three steps. 
%\end{example}
%\begin{example}[\cite{FKN17tocl}]
\label{exa:factlctrs}
To implement an LCTRS calculating the \emph{factorial} function, we use the signature $\Sigma$ above %from Example~\ref{exa:factsignature}
 and the following rules:
$%\[
\cR_{\symb{fact}}
 = \{\ 
\CRule{\symb{fact}(x)}{\symb{1}}{x \leq \symb{0}},
~~%~~
\CRule{\symb{fact}(x)}{x \times \symb{fact}(x-\symb{1})}{\neg (x \leq \symb{0})}
\ \}
$. %\]
Expected starting terms are, e.g.,
$\symb{fact}(\symb{42})$ or $\symb{fact}(\symb{fact}(\symb{-4}))$.
%Note that ground terms are fully built using symbols in $\Sigmaterms$.
Using the constrained rewrite rules in $\cR_{\symb{fact}}$, $\symb{fact}(\symb{3})$ reduces in ten steps to $\symb{6}$.
\end{example}

\subsection{{\SIMP}: a Small Imperative Language with Global Variables and Function Calls}

In this section, we recall the syntax of {\SIMPo}, a small imperative language (cf.~\cite{Fer14}). %\emph{while programs} (see e.g.,~\cite{Rey98}).
To deal with global variables and function calls, we add them into the ordinary syntax and semantics of {\SIMPo} in a natural way.
We refer to such an extended language as {\SIMP}.

We first show the syntax adopting a C-like notation. 
A \emph{program} $\lrang{P}$ of {\SIMP} is defined by the following BNF:
\begin{eqnarray*}
\lrang{P} &::=& \lrang{D}~~\lrang{F} \\
\lrang{D} & ::= & \epsilon \mid  \Assign{\IntVarDecl{v}}{n};~\lrang{D} \\
\lrang{F} & ::= & \epsilon \mid  \IntFunDecl{f}{\IntVarDecl{x_1},\ldots,\IntVarDecl{x_m}}~=~\{~\lrang{D}~~\lrang{S}~~\Return{E};~\}~~\lrang{F} \\
\lrang{S} &::=& \Skip
					\mid \Assign{v}{\lrang{E}}\,;~\lrang{S} 
					\mid \Assign{v}{\FCall{f}{\lrang{E},\ldots,\lrang{E}}}\,;~\lrang{S}
%					\mid \lrang{S}\,; \lrang{S} 
                    \mid \If{\lrang{B}}{\lrang{S}}{\lrang{S}}~\lrang{S}
%                          \mid \while\, @\,\lrang{B}\,(\lrang{B}) \{ \lrang{comm} \} 
                          \mid \While{\lrang{B}}{\lrang{S}}~\lrang{S} \\
\lrang{E} &::=& n % s0 \mid 1 \mid 2 \mid ...\\
%                      & \mid & 
                     \mid v % \lrang{var} 
%                     \mid (-\lrang{E})
					\mid (\lrang{E} \mathrel{+} \lrang{E}) 
					\mid (\lrang{E} \mathrel{-} \lrang{E}) 
%					\mid (\lrang{E} \mathrel{*} \lrang{E}) 
%					\mid (\lrang{E} \mathrel{/} \lrang{E}) 
%					\mid \FCall{f}{\lrang{E},\ldots,\lrang{E}}
						\\
\lrang{B} &::=& \True \mid \False
                         \mid (\lrang{E} \EQ \lrang{E})
%                         \mid (\lrang{E} \neq \lrang{E}) 
                        \mid (\lrang{E} < \lrang{E})
                        \mid (\lnot \lrang{B})
%                         \mid (\lrang{B} \land \lrang{B})
                         \mid (\lrang{B} \lor \lrang{B}) 
\end{eqnarray*}
where $n \in \Int$, $v \in \cV$, $f$ is a function name, and we may omit brackets in the usual way.
The empty sequence ``$\Skip$'' is used instead of the ``skip'' command.
%C does not have boolean values, but we use $\True$ and $\False$ for readability.
To simplify discussion, we do not use other operands such as multiplication and division, but 
we use $\ne$, $\leq$, $>$, $\geq$, $\land$, $\implies$, etc, as syntactic sugars.
We also use the $\texttt{for}$-statement as a syntactic sugar.
We assume that a function name $f$ has a fixed arity, and the definition and call of $f$ are consistent with the arity.
A program $P$ consists of declarations of \emph{global variables} (with initialization) and functions.
For a program $P$, we denote the set of global variables appearing in $P$ by $\GVar(P)$:
% and the initial value of a global variable $x$ of $P$ is denoted by $\InitialVal{x}$:
let $P$ be $\Assign{\IntVarDecl{x_1}}{n_1};\ldots;\Assign{\IntVarDecl{x_k}}{n_k};\IntFunDecl{f}{\ldots}~=~\{\ldots\}~\ldots$, then
$\GVar(P)=\{x_1,\ldots,x_n\}$.
% and $\InitialVal{x_i}=n_i$ for all $1 \leq i \leq k$.
We assume that each function $f$ is defined at most once in a program $P$ and any function called in a function defined in $P$ is defined in $P$.
To simplify the semantics, we assume that local variables in function declarations are different from global variables and parameters of functions.
%Given an assignment $\theta$ for $\FVar(P)$, we write $\theta \mathrel{\Eval{P}} \theta'$ if the execution of $P$ starts with $\theta$ and halts with an assignment $\theta'$.
An \emph{assignment} is defined by a substitution whose range is over the integers, which may be used for terms in the setting of LCTRSs.
%We abuse such an assignment as a substitution for terms in the setting of LCTRSs.
We deal with {\SIMP} programs that can be successfully compiled as C programs.

\begin{example}
The program $\PrgrmA$ in Figure~\ref{fig:sum} is a {\SIMP} program, and we have that $\GVar(\PrgrmA) = \{ \var{num} \}$.
% and $\InitialVal[\PrgrmA]{\var{num}}=0$.
\end{example}

\begin{figure}[t]
%\begin{verbatim}
%  int num = 0;
%
%  int sum(int x){
%    num = num + 1;
%    if( x <= 0 ){
%      return 0;
%    }else{
%      return x + sum(x - 1);
%    }
%  }
%
%  int main(){
%    int n = 3;
%    sum(n);
%    return 0;
%  }
%\end{verbatim}
\begin{multicols}{2}
\begin{verbatim}
  int num = 0;

  int sum(int x){
    int z = 0;
    num = num + 1;
    if( x <= 0 ){
      z = 0;
    }else{
      z = sum(x - 1);
      z = x + z;
    }
    return z;
  }
\end{verbatim}
\columnbreak
\begin{verbatim}
  int main(){
    int z = 3;
    z = sum(z);
    return 0;
  }
\end{verbatim}
\end{multicols}
\vspace{-5pt}
\caption{a {\SIMP} program $\PrgrmA$ obtained by adding the definition of \texttt{main} into the program for \texttt{sum}.}
\label{fig:sum}
\end{figure}

The semantics $\toExpr$ of integer and boolean expressions is defined as usual (see Figure~\ref{fig:E-semantics}):
given an expression $e$ and an assignment $\sigma$ with $\Dom(\sigma) \supseteq \FVar(e)$, we write $(e,\sigma) \toExpr v$ where $v$ is the resulting value obtained by evaluating $e$ with $\sigma$.
The \emph{transition system} defining the semantics of a {\SIMP} program $P$ is defined by
\begin{itemize}
	\item \emph{configurations} of the form $\Config{\alpha}{\sigma_0}{\sigma_1}$, where 
	\begin{itemize}
		\item $\alpha$ is of the form ``$\delta~\beta$'' with variable declarations $\delta$,%
		\footnote{ Variable declarations $\delta$ may be the empty sequence.}
		 and a statement $\beta$, and
		\item $\sigma_0,\sigma_1$ are assignments for global and local variables, respectively, which are represented by partial functions from variables to integers---the \emph{update} $\Update{\sigma}{x}{n}$ of an assignment $\sigma$ w.r.t.\ $x$ for an integer $n$ is defined as follows:
			if $x=y$ then $\Update{\sigma}{x}{n}(y) = n$, and otherwise, $\Update{\sigma}{x}{n}(y) = \sigma(y)$,
%			\[
%			\Update{\sigma}{x}{n}(y) = 
%			\left\{
%			\begin{array}{ll}
%				n & \mbox{if $x=y$} \\
%				\sigma(y) & \mbox{otherwise} \\
%			\end{array}
%			\right.
%			\]
	\end{itemize}
	and
	\item a \emph{transition relation} $\Downarrow_P$ between configurations, which is defined as a \emph{big-step} semantics by the inference rules illustrated in Figure~\ref{fig:S-semantics}.
\end{itemize}
%The partial function that is not defined for any variable is denoted by $\emptyset$.
We assume that for any configuration $\Config{\alpha}{\sigma_0}{\sigma_1}$ for a program $P$, the assignment $\sigma_0$ is defined for all global variables of $P$.
%Assume that $\GVar(P)=\{x_1,\ldots,x_k\}$.
To compute the result of a function call $\FCall{f}{e_1,\ldots,e_m}$ under assignments $\sigma_0,\sigma_1$ for $\GVar(P)$ and $\FVar(e_1,\ldots,e_m)\setminus\GVar(P)$, given a fresh variable $x$, we start with the configuration 
$
\Config{\Assign{x}{\FCall{f}{e_1,\ldots,e_m}}}{\sigma_0}{\Update{\sigma_1}{x}{0}}
$.
When $\Config{\Assign{x}{\FCall{f}{e_1,\ldots,e_m}}}{\sigma_0}{\Update{\sigma_1}{x}{0}} \Downarrow_P \Config{\Skip}{\sigma_0'}{\sigma_1'}$ holds, the execution halts and the result of the function call $\FCall{f}{e_1,\ldots,e_m}$ under $\sigma_0,\sigma_1$ is $\sigma_1'(x)$.

\begin{figure}[p]
\begin{tabular}{@{}c@{~~}c@{}}
	$\Inference{n \in \Int}{(n,\sigma) \toExpr n}$
&
	$\Inference{x \in \cV}{(x,\sigma) \toExpr \sigma(x)}$
\\[12pt]
%	$\Inference{(e_1,\sigma) \toExpr n_1 ~~~~ (e_2,\sigma) \toExpr n_2 ~~~~ n_1 + n_2 = n \in \Int}{(e_1 + e_2,\sigma) \toExpr n}$
%&
%	$\Inference{(e_1,\sigma) \toExpr n_1 ~~~~ (e_2,\sigma) \toExpr n_2 ~~~~ n_1 - n_2 = n \in \Int }{(e_1 - e_2,\sigma) \toExpr n}$
%\\[12pt]
\multicolumn{2}{c}{%
	$\Inference{(e_1,\sigma) \toExpr n_1 ~~~~ (e_2,\sigma) \toExpr n_2 ~~~~ n_1 \bowtie n_2 = n \in \Int ~~~~ {\bowtie} \in \{{+},{-}\}}{(e_1 \bowtie e_2,\sigma) \toExpr n}$
}
\\[12pt]
	$\Inference{(e_1,\sigma) \toExpr n_1 ~~~~ (e_2,\sigma) \toExpr n_2 ~~~~ n_1 = n_2}{(e_1 \EQ e_2,\sigma) \toExpr \True}$
&
	$\Inference{(e_1,\sigma) \toExpr n_1 ~~~~ (e_2,\sigma) \toExpr n_2 ~~~~ n_1 \ne n_2}{(e_1 \EQ e_2,\sigma) \toExpr \False}$
\\[12pt]
	$\Inference{(e_1,\sigma) \toExpr n_1 ~~~~ (e_2,\sigma) \toExpr n_2 ~~~~ n_1 < n_2}{(e_1 < e_2,\sigma) \toExpr \True}$
&
	$\Inference{(e_1,\sigma) \toExpr n_1 ~~~~ (e_2,\sigma) \toExpr n_2 ~~~~ n_1 \geq n_2}{(e_1 < e_2,\sigma) \toExpr \False}$
\\[12pt]
	$\Inference{(\varphi,\sigma) \toExpr \False}{(\neg \varphi,\sigma) \toExpr \True}$
&
	$\Inference{(\varphi,\sigma) \toExpr \True}{(\neg \varphi,\sigma) \toExpr \False}$
\\[12pt]
	$\Inference{(\varphi_1,\sigma) \toExpr b_1 ~~~~ (\varphi_2,\sigma) \toExpr b_2 ~~~~ \True \in \{b_1,b_2\}}{(\varphi_1 \vee \varphi_2),\sigma) \toExpr \True}$
&
	$\Inference{(\varphi_1,\sigma) \toExpr \False ~~~~ (\varphi_2,\sigma) \toExpr \False }{(\varphi_1 \vee \varphi_2),\sigma) \toExpr \False}$	
\\
\end{tabular}
\caption{the inference rules for the semantics of {\SIMP} expressions.}
\label{fig:E-semantics}
\end{figure}

\begin{figure}[p]
\noindent
\begin{tabular}{@{}c@{}}
$
	\Inference{}{\Config{\Skip}{\sigma_0}{\sigma_1} \Downarrow_P \Config{\Skip}{\sigma_0}{\sigma_1}}
$
\qquad
$
	\Inference{\Config{\beta}{\sigma_0}{\Update{\sigma_1}{x}{n}} \Downarrow_P \Config{\Skip}{\sigma_0'}{\sigma_1'}}{\Config{\Assign{\IntVarDecl{x}}{n};~\beta}{\sigma_0}{\sigma_1} \Downarrow_P \Config{\Skip}{\sigma_0'}{\sigma_1'}}
$
\\[12pt]
$
	\Inference{(e,\sigma_0\cup\sigma_1) \toExpr n ~~~~ x \in \GVar(P) ~~~~ \Config{\beta}{\Update{\sigma_0}{x}{n}}{\sigma_1} \Downarrow_P \Config{\Skip}{\sigma_0'}{\sigma_1'}}{\Config{\Assign{x}{e};~\beta}{\sigma_0}{\sigma_1} \Downarrow_P \Config{\Skip}{\sigma_0'}{\sigma_1'}}
$
\\[12pt] %\qquad
$
	\Inference{(e,\sigma_0\cup\sigma_1) \toExpr n ~~~~ x \notin \GVar(P) ~~~~ \Config{\beta}{\sigma_0}{\Update{\sigma_1}{x}{n}} \Downarrow_P \Config{\Skip}{\sigma_0'}{\sigma_1'}}{\Config{\Assign{x}{e};~\beta}{\sigma_0}{\sigma_1} \Downarrow_P \Config{\Skip}{\sigma_0'}{\sigma_1'}}
$
\\[12pt]
%$
%	\Inference{\Config{\alpha_1}{\sigma_0}{\sigma_1} \Downarrow_P \Config{\Skip}{\sigma_0'}{\sigma_1'} ~~~~ \Config{\alpha_2}{\sigma_0'}{\sigma_1'} \Downarrow_P \Config{\Skip}{\sigma_0''}{\sigma_1''}}{\Config{\alpha_1;\alpha_2}{\sigma_0}{\sigma_1} \Downarrow_P \Config{\Skip}{\sigma_0''}{\sigma_1''}}
%$
%\\[12pt]
$
	\Inference{(\varphi,\sigma_0\cup\sigma_1) \toExpr \True ~~~~ \Config{\alpha_1~\beta}{\sigma_0}{\sigma_1} \Downarrow_P \Config{\Skip}{\sigma_0'}{\sigma_1'}}{\Config{\If{\varphi}{\alpha_1}{\alpha_2}~\beta}{\sigma_0}{\sigma_1} \Downarrow_P \Config{\Skip}{\sigma_0'}{\sigma_1'}}
$
\\[12pt]
$
	\Inference{(\varphi,\sigma_0\cup\sigma_1) \toExpr \False ~~~~ \Config{\alpha_2~\beta}{\sigma_0}{\sigma_1} \Downarrow_P \Config{\Skip}{\sigma_0'}{\sigma_1'}}{\Config{\If{\varphi}{\alpha_1}{\alpha_2}~\beta}{\sigma_0}{\sigma_1} \Downarrow_P \Config{\Skip}{\sigma_0'}{\sigma_1'}}
$
\\[12pt]
$
	\Inference{(\varphi,\sigma_0\cup\sigma_1) \toExpr \True ~~~~ \Config{\alpha}{\sigma_0}{\sigma_1} \Downarrow_P \Config{\Skip}{\sigma_0'}{\sigma_1'} ~~~~ \Config{\While{\varphi}{\alpha}~\beta}{\sigma_0'}{\sigma_1'} \Downarrow_P \Config{\Skip}{\sigma_0''}{\sigma_1''}}{\Config{\While{\varphi}{\alpha}~\beta}{\sigma_0}{\sigma_1} \Downarrow_P \Config{\Skip}{\sigma_0''}{\sigma_1''}}
$
\\[12pt]
$
	\Inference{(\varphi,\sigma_0\cup\sigma_1) \toExpr \False ~~~~ \Config{\beta}{\sigma_0}{\sigma_1} \Downarrow_P \Config{\Skip}{\sigma_0'}{\sigma_1'}}{\Config{\While{\varphi}{\alpha}~\beta}{\sigma_0}{\sigma_1} \Downarrow_P \Config{\Skip}{\sigma_0'}{\sigma_1'}}
$
\\[12pt]
$
	\Inference{\! \forall i.\ %\in\{1,\ldots,m\}.\ 
	(e_i,\sigma_0\cup\sigma_1) \!\toExpr\! n_i ~~~ \Config{\alpha}{\sigma_0}{\sigma_2} \!\Downarrow_P\! \Config{\Skip}{\sigma_0'}{\sigma_1'} ~~~ (e,\sigma_0'\cup\sigma_1') \!\toExpr\! n ~~~ \Config{\beta}{\sigma_0''}{\sigma_1''} \!\Downarrow_P\! \Config{\Skip}{\sigma_0'''}{\sigma_1'''} \!}{\Config{\Assign{x}{\FCall{f}{e_1,\ldots,e_m}};~\beta}{\sigma_0}{\sigma_1} \Downarrow_P \Config{\Skip}{\sigma_0'''}{\sigma_1'''}}
$
\\
\end{tabular}
	where 
	\begin{itemize}
		\item $\IntFunDecl{f}{\IntVarDecl{y_1},\ldots,\IntVarDecl{y_m}}~=~\{~\alpha~~\Return{e};~\}$ is in $P$, 
		\item $\sigma_2 = \{y_1\mapsto n_1,\ldots,y_m\mapsto n_m\}$,
		\item if $x \in \GVar(P)$ then $\sigma_0''=\Update{\sigma_0'}{x}{n}$, and otherwise $\sigma_0''=\sigma_0'$,
			and
		\item if $x \in \GVar(P)$ then $\sigma_1''=\sigma_1$, and otherwise $\sigma_1''=\Update{\sigma_1}{x}{n}$
	\end{itemize}
%\vspace{-5pt}
\caption{the inference rules for the semantics of {\SIMP} statements and variable-declarations.}
\label{fig:S-semantics}
\end{figure}

%\begin{example}
%The execution of the program $\PrgrmA$ in Figure~\ref{fig:sum} (i.e, the result of computing the configuration $\Config{\Assign{y}{\FCall{\mathtt{main}}{}}}{\{\var{num} \mapsto 0\}}{\emptyset}$) is illustrated in Figure~\ref{fig:sum-execution}.
%\end{example}
%
%\begin{figure}[t]
%\footnotesize
%\[
%\Inference{%
%\Inference{%
%\Inference{%
%(n,\{\var{num}\mapsto 0, ~ n \mapsto 0 \}) \toExpr 0
%\quad
%\Config{\Assign{\IntVarDecl{z}}{0};~\ldots	}{\{\var{num} \mapsto 0\}}{\{ x \mapsto 0, ~ z \mapsto 0\}} \Downarrow_{\PrgrmA} \Config{\Skip}{\{\var{num} \mapsto ? \}}{\{x \mapsto 0, ~ z \mapsto ? \}}
%\quad
%(z
%}{%
%\Config{\Assign{n}{\FCall{\mathtt{sum}}{n}};~\ldots}{\{\var{num} \mapsto 0\}}{\{n \mapsto 3\}} \Downarrow_{\PrgrmA} \Config{\Skip}{\{\var{num} \mapsto ? \}}{\{n \mapsto ?\}}
%}}{%
%\Config{\Assign{\IntVarDecl{n}}{3};~\ldots}{\{\var{num} \mapsto 0\}}{\{n \mapsto 3\}} \Downarrow_{\PrgrmA} \Config{\Skip}{\{\var{num} \mapsto ? \}}{\{n \mapsto ?\}}
%}
%\quad
%(0,\{\var{num}\mapsto ?, ~ n \mapsto ?\}) \toExpr 0
%}{%
%\Config{\Assign{y}{\FCall{\mathtt{main}}{}}}{\{\var{num} \mapsto 0\}}{\emptyset} \Downarrow_{\PrgrmA} \Config{\Skip}{\{\var{num} \mapsto ? \}}{\{y \mapsto 0\}}
%}
%\]
%\caption{the execution for the configuration $\Config{\Assign{y}{\FCall{\mathtt{main}}{}}}{\{\var{num} \mapsto 0\}}{\emptyset}$.}
%\label{fig:sum-execution}
%\end{figure}

%%%%%%%%%%%%%%%%%%%%%%%%%%%%%%%%%%%%%%%%%%%%%%%%%%%%%%%%%%%%%%%%%%%%%%%%%
\section{A New Approach to Transformations of Imperative Programs}
% with Global Variables}
\label{sec:new-transformation}
%%%%%%%%%%%%%%%%%%%%%%%%%%%%%%%%%%%%%%%%%%%%%%%%%%%%%%%%%%%%%%%%%%%%%%%%%

In this section, using an example, we introduce a new approach to transformations of imperative programs with function calls and global variables.

\subsection{The Existing Transformation of Functions Accessing Global Variables}
\label{subsec:existing-transformation}

In this section, we briefly recall the transformation of imperative programs with functions accessing global variables~\cite{FKN17tocl} using the program $\PrgrmA$ in Figure~\ref{fig:sum}. %, which is simpler than that in Section~\ref{sec:intro}.
Unlike $\cRseq$ in Section~\ref{sec:intro}, in the following, we do not optimize generated rewrite rules in LCTRSs in order to make it easier to understand how to precisely transform programs.
The program $\PrgrmA$ is transformed into the following LCTRS with the sort set $\{\INT,\BOOL,\sort{state}\}$ and the standard integer signature $\Sigmaint$~\cite{FKN17tocl}:
\[
	\cRsumrec =
 \left\{
 \begin{array}{r@{\>}c@{\>}lr@{\,}c@{\,}l}
 	\Rule{\symb{sum}(x,\var{num}) &}{& \SumA(x,\var{num},\symb{0})}, \\
 	\Rule{\SumA(x,\var{num},z) &}{& \SumB(x,\var{num}+\symb{1},z)}, \\
 	\CRule{\SumB(x,\var{num},z) &}{& \SumC(x,\var{num},z) &}{& x \leq \symb{0} &}, \\
 	\CRule{\SumB(x,\var{num},z) &}{& \SumE(x,\var{num},z) &}{& \neg (x \leq \symb{0}) &}, \\
	\Rule{\SumC(x,\var{num},z) &}{& \SumD(x,\var{num},\symb{0})}, \\
 	\Rule{\SumD(x,\var{num},z) &}{& \SumI(x,\var{num},z)}, \\
	\Rule{\SumE(x,\var{num},z) &}{& \SumF(x,\var{num},z,\symb{sum}(x-\symb{1}))}, \\
 	\Rule{\SumF(x,\var{num_{old}},z,\symb{return}(y,\var{num_{new}})) &}{& \SumG(x,\var{num_{new}},y)}, \\
 	\Rule{\SumG(x,\var{num},z) &}{& \SumH(x,\var{num},x+z)}, \\
 	\Rule{\SumH(x,\var{num},z) &}{& \SumI(x,\var{num},z)}, \\
 	\Rule{\SumI(x,\var{num},z) &}{& \symb{return}(z,\var{num})}, \\[5pt]
 	
 	\Rule{\symb{main}(\var{num}) &}{& \MainA(\var{num},\symb{3})}, \\
 	\Rule{\MainA(\var{num},z) &}{& \MainB(\var{num},z,\symb{sum}(z,\var{num}))}, \\
 	\Rule{\MainB(\var{num_{old}},z,\symb{return}(y,\var{num_{new}})) &}{& \MainC(y,\var{num_{new}})}, \\
 	\Rule{\MainC(\var{num},z) &}{& \symb{return}(\symb{0},\var{num})} \\
 \end{array}
 \right\}
\]
where 
$\symb{main}: \INT \funtype \sort{state}$,
$\SumA,\SumB,\SumC,\SumD,\SumE,\SumG,\SumH,\SumI: \INT \times \INT \times \INT \funtype \sort{state}$, 
$\symb{sum},\MainA,\MainC,\symb{return}: \INT \times \INT \funtype \sort{state}$,
$\MainB: \INT \times \INT \times \sort{state} \funtype \sort{state}$,
and
$\SumF: \INT \times \INT \times \INT \times \sort{state} \funtype \sort{state}$.
%First, we explain how to transform of C program with function calls into LCTRS. 
The declaration of local variable $\mathtt{z}$ of $\mathtt{sum}$ is represented by the first rule of $\cRsumrec$, which stores the initial value $\symb{0}$ in the third argument of $\SumA$.
The \texttt{if}-statement is represented by rules of $\SumB$, $\SumD$, and $\SumH$;
The first rule of $\SumB$ enters the body of the \texttt{then}-statement if $x \leq \symb{0}$ holds, and the second rule of $\SumB$ enters the body of the \texttt{else}-statement if $x \leq \symb{0}$ does not hold (i.e., $\neg (x \leq \symb{0})$ holds);
The end of the \texttt{if}-statement is represented by terms rooted by $\SumI$, and the rules of $\SumD$ and $\SumH$ are used to exit the bodies of the \texttt{then}- and \texttt{else}-statements, respectively.

To represent the function call \verb+sum(x - 1)+, the auxiliary function symbol $\SumF$ takes the term $\symb{sum}(x-\symb{1},\var{num})$ as the fourth argument.
The function symbol $\symb{sum}$ takes two arguments, while the original function $\mathtt{sum}$ in the program takes one argument.
This is because the global variable $\mathtt{num}$ is accessed during the execution of $\mathtt{sum}$, and we pass the value stored in $\var{num}$ to $\symb{sum}$, passing the variable itself to $\symb{sum}$ in the constructed rule.
The rule of $\SumA$ increments the global variable $\mathtt{num}$, and thus, we include the value stored in $\var{num}$ in the result of $\symb{sum}$ by means of $\symb{return}(z,\var{num_{new}})$.
The rule of $\SumF$ is used after the reduction of $\symb{sum}(x-\symb{1},\var{num})$, receiving the result by means of the pattern $\symb{return}(y,\var{num_{new}})$.
The updated value stored in $\mathtt{num}$ is received by $\var{num_{new}}$, and the rule of $\SumF$ updates the global variable $\mathtt{num}$ by passing $\var{num_{new}}$ to the second argument of $\SumG$.
We do the same for the function call \verb+sum(z)+ in the auxiliary function symbol $\MainB$.
For the execution of the program, we have the reduction of $\cRsumrec$ illustrated in Figure~\ref{fig:sumrec-reduction}.
Note that the global variable $\mathtt{num}$ is initialized by $\symb{0}$ and we started from $\symb{main}(\symb{0})$.
From the reduction, we can see that the called function is the only running one under sequential execution, and others are waiting for the called function halting.
The approach above to function calls and global variables is enough for sequential execution.

\begin{figure}[t]
\[
\begin{array}{@{}l@{\>}l@{}}
	\symb{main}(\symb{0}) 
	& \to_{\cRsumrec}
		\MainA(\symb{0},\symb{3}) \\
	& \to_{\cRsumrec}
		\MainB(\symb{0},\symb{3},\symb{sum}(\symb{3},\symb{0})) \\
	& \to_{\cRsumrec}
		\MainB(\symb{0},\symb{3},\SumA(\symb{3},\symb{0},\symb{0})) \\
	& \to_{\cRsumrec}
		\MainB(\symb{0},\symb{3},\SumB(\symb{3},\symb{0}+\symb{1})) \\
	& \to_{\cRsumrec}
		\MainB(\symb{0},\symb{3},\SumB(\symb{3},\symb{1},\symb{0})) \\
	& \to_{\cRsumrec}
		\MainB(\symb{0},\symb{3},\SumE(\symb{3},\symb{1},\symb{0})) \\
	& \to_{\cRsumrec}
		\MainB(\symb{0},\symb{3},\SumF(\symb{3},\symb{1},\symb{0},\symb{sum}(\symb{3}-\symb{1},\symb{1}))) \\
	& \to_{\cRsumrec}
		\MainB(\symb{0},\symb{3},\SumF(\symb{3},\symb{1},\symb{0},\symb{sum}(\symb{2},\symb{1}))) \\

%%%%%%%%%%%%%%%%%%%%%
	& \to_{\cRsumrec} \cdots \\

	& \to_{\cRsumrec}
		\MainB(\symb{0},\symb{3},\SumG(\symb{3},\symb{4},\symb{3}))) \\
	& \to_{\cRsumrec}
		\MainB(\symb{0},\symb{3},\SumH(\symb{3},\symb{4},\symb{3}+\symb{3}))) \\
	& \to_{\cRsumrec}
		\MainB(\symb{0},\symb{3},\SumH(\symb{3},\symb{4},\symb{6}))) \\
	& \to_{\cRsumrec}
		\MainB(\symb{0},\symb{3},\SumI(\symb{3},\symb{4},\symb{6})) \\
	& \to_{\cRsumrec}
		\MainB(\symb{0},\symb{3},\symb{return}(\symb{6},\symb{4})) \\
	& \to_{\cRsumrec}
		\MainC(\symb{6},\symb{4})\\
	& \to_{\cRsumrec}
		\symb{return}(\symb{0},\symb{4}) \\
\end{array}
\]
%\vspace{-10pt}
\caption{the reduction of $\cRsumrec$ for the execution of the program for \texttt{sum}.}
\label{fig:sumrec-reduction}
\end{figure}

%Next, we explain the treatment for global variables in previous research. Only the sequential execution is assumed, and a value stored in the global variable is passed to a function call as an extra argument.  after receiving from the called function a value of the global variable that may be updated in executing the function call, restoring the value in the global variable. Consider the following example.
%\begin{verbatim}
%    int num = 0;
%
%    int sum(int x){
%      num ++;
%      if( x <= 0 ) return 0;
%      else return x + sum(x-1);
%    }
%\end{verbatim}
%When transforming this program by the original method, the following rule of LCTRS is generated~\cite{FKN17tocl}:
%\[
% \left\{
% \begin{array}{r@{\>}c@{\>}lr@{\,}c@{\,}l}
% 	\Rule{\symb{sum}(x,num) &}{& \SumB(x,num+\symb{1})} \\
% 	\CRule{\SumB(x,num) &}{& \symb{return}(x,num) &}{& x \leq \symb{0} &} \\
% 	\CRule{\SumB(x,num) &}{& \SumC(x,num,\symb{sum}(x-\symb{1},num)) &}{& \neg (x \leq \symb{0}) &} \\
% 	\Rule{\SumC(x,num,\symb{return}(y,num')) &}{& \symb{return}(x + y, num')} \\
% \end{array}
% \right\}
%\]

In the LCTRS $\cRsumrec$ above, the function symbol $\SumF$ recursively calls $\symb{sum}$ in its fourth argument.
For this reason, the running function is located below $\SumF$, and positions where $\symb{sum}$ is called are not unique.
The above approach to transform function calls is very naive but not so general.
For example, to model parallel execution, a value stored in a global variable does not have to be passed to a particular function or a process because another function or process may access the global variable.
%In the next section, we propose a new way of handling global variables.
%Therefore, the called function is placed inside the term, and positions of function calls in terms are not unique, and a function is called at a very deep position of a term when the number of nesting increases. In this transformation, a value of the global variable such as num is sequentially passed to called functions, and it is updated to the new value when receiving the result. In the sequential execution, there is no problem with this transformation method, but it is not suitable for modeling parallel execution. In order to solve this problem, we propose a new transformation method in the next section.

\subsection{Another Approach to Global Variables}

In this section, we show another approach to the treatment of global variables.

%In the parallel execution, global variables and shared memory are accessed from multiple function or thread.
%Therefore, it is not possible to adopt the original approach that pass values of global variables when calling a function. 
To adapt to more general settings such as parallel execution, global variables used like shared memories should be located at fixed addresses (i.e., fixed positions of terms) because they may be accessed from two or more functions or processes.
%Since a value stored in a global variable may always is referenced, the value is placed in a fixed position and don't move it. 
To keep values stored in global variables at fixed positions, we do not pass (values of) global variables to called functions in order to avoid locally updating global variables.
%In order to realize this, 
To this end, we prepare a new function symbol $\symb{env}$ to represent the whole environment for execution, and make $\symb{env}$ have values stored in global variables in its arguments.
 %At this time, it is necessary to review how to express the function call. 
 In addition, we make $\symb{env}$ have an extra argument where functions or processes are executed sequentially.%
 \footnote{
 When we execute $n$ ($>1$) processes in parallel, we make $\symb{env}$ have $n$  extra arguments where the $i$-th process is executed in the $i$-th extra argument.
 }
  For example, the process of executing the above program is expressed as follows:
\[
\symb{env}(\symb{0},\symb{main}())
\]
Note that $\symb{env}$ has the sort $\INT \times \sort{state} \funtype \sort{env}$, where $\sort{env}$ is a new sort for environment.
The first argument of $\symb{env}$ is the place where values for the global variable \texttt{num} are stored, and the second argument of $\symb{env}$ is the place where functions are executed, e.g., the main function \texttt{main} is called as in the above term.

We do not change the transformation of \emph{local statements}---statements without accessing global variables---in function definitions.
Let us consider the execution of the program, i.e., \texttt{main}.
All the statements in \texttt{main} and the first statement of $\mathtt{sum}$ are local, and thus, we transform the definition of \texttt{main} as well as $\cRsumrec$:
\[
 \left\{
 \begin{array}{r@{\>}c@{\>}lr@{\,}c@{\,}l}
 	\Rule{\symb{main}() &}{& \MainA(\symb{3})}, \\
 	\Rule{\MainA(z) &}{& \MainB(z,\symb{sum}(z))}, \\
 	\Rule{\MainB(z,\symb{return}(y)) &}{& \MainB(y)}, \\
 	\Rule{\MainC(z) &}{& \symb{return}(\symb{0})}, \\[5pt]
	\Rule{\symb{sum}(x) &}{& \SumA(x,\symb{0})} \\
 \end{array}
 \right\}
\]
The symbol $\symb{return}$ no longer contains values for the global variable $\mathtt{num}$.
In executing the program (i.e., \texttt{main}), the first access to the global variable $\mathtt{num}$ is the statement ``\verb|num = num + 1|''
%, the  increment of \texttt{num}, 
in the definition of $\symb{sum}$.
The initial term $\symb{env}(\symb{0},\symb{main}())$ can be reduced to $\symb{env}(\symb{0},\MainB(\symb{3},\SumA(\symb{3},\symb{0})))$, and thus, the first execution of the statement ``\verb|num = num + 1|'' can be expressed by the following rewrite rule for $\symb{env}$:
\[
%\Rule{\symb{env}(num,\symb{sum}(x))}{\symb{env}(num+\symb{1},\SumB(x))}
\Rule{\symb{env}(\var{num},\MainB(z_0,\SumA(x,z)))}{\symb{env}(\var{num}+\symb{1},\MainB(z_0,\SumB(x,z)))}
\]
%Next, to represent the \texttt{if} statement, the \texttt{return} statement, and the recursive call using the original transformation method, the following rules are obtained.
The other statements in the definition of \texttt{sum} are local and we transform them into the following rules, as well as $\cRsumrec$:
\[
 \left\{
 \begin{array}{r@{\>}c@{\>}lr@{\,}c@{\,}l}
 	\CRule{\SumB(x,z) &}{& \SumC(x,z) &}{& x \leq \symb{0} &}, \\
 	\CRule{\SumB(x,z) &}{& \SumE(x,z) &}{& \neg (x \leq \symb{0}) &}, \\
	\Rule{\SumC(x,z) &}{& \SumD(x,\symb{0})}, \\
 	\Rule{\SumD(x,z) &}{& \SumI(x,z)}, \\
	\Rule{\SumE(x,z) &}{& \SumF(x,z,\symb{sum}(x-\symb{1}))}, \\
 	\Rule{\SumF(x,z,\symb{return}(y)) &}{& \SumG(x,y)}, \\
 	\Rule{\SumG(x,z) &}{& \SumH(x,x+z)}, \\
 	\Rule{\SumH(x,z) &}{& \SumI(x,z)}, \\
 	\Rule{\SumI(x,z) &}{& \symb{return}(z)} \\
 \end{array}
 \right\}
\]
%With only the above rules, we cannot represent the execution of recursively called $\symb{sum}$ for the first time. 
Unfortunately, the above rules are not enough to capture all possible executions, e.g. the second execution of ``\verb|num = num + 1|'', which is done by the second call of \texttt{sum}, is not expressed yet.
Thus, we prepare the following rule:
\[
%\Rule{\symb{env}(num,\SumC(x',\symb{sum}(x)))}{\symb{env}(num+\symb{1},\SumC(x',\SumB(x)))}
\Rule{\symb{env}(\var{num},\MainB(z_0,\SumF(x',z',\SumA(x,z))))}{\symb{env}(\var{num}+\symb{1},\MainB(z_0,\SumF(x',z',\SumB(x,z))))}
\]
%In addition, to represent the execution of $\symb{sum}$ recursively recalled for the second time, we prepare the following rule.
In addition, \texttt{sum} is further recursively called, and we need the following rule:
\[
%\Rule{\symb{env}(num,\SumC(x'',\SumC(x',\symb{sum}(x))))}{\symb{env}(num+\symb{1},\SumC(x'',\SumC(x',\SumB(x))))}
\Rule{\symb{env}(\var{num},\MainB(z_0,\SumF(x',z',\SumF(x'',z'',\SumA(x,z)))))}{\symb{env}(\var{num}+\symb{1},\MainB(z_0,\SumF(x',z',\SumF(x'',z'',\SumB(x,z)))))}
\]
In summary, we need similar rules for all recursive calls of $\mathtt{sum}$.
The function \texttt{sum} may receive all the (finitely many) integers, and we need many similar rules, all of which express the increment of \texttt{num}.
In addition, we may need other rules for the case where we add other functions calling \texttt{sum} into the program.
More generally, the nesting of function calls cannot be fixed, and thus, along the above approach, we may need infinitely many rewrite rules.
This means that the above approach is not adequate for recursive functions.
%
%Since the depth of sum to be recursively called is not constant, it is necessary to prepare a similar rule in a large amount or indefinitely to represent all the execution. 

The troublesome observed by means of $\PrgrmA$ is caused by the fact that positions where $\symb{sum}$ is called are not unique in the above approach.
We will show another approach to avoid this troublesome in the next section.
%In the next section, we propose a new representation method of function call that solve this problem.

\subsection{Using a Call Stack for Function Calls}

In this section, using $\PrgrmA$ in Figure~\ref{fig:sum}, we show a new representation of function calls for LCTRSs.

%In transformation of previous section, we may need (possibly infinitely) many rules for a transition related to a global variable. 
The approach to the treatment of global variables in the previous section needs finitely or infinitely many similar rules for statements accessing global variables, and we have to add other similar rules when we introduce another function that may call itself or other functions. 
As described at the end of the previous section, the cause of this problem is that positions where functions are called in terms rooted by $\symb{env}$ are not unique due to nestings of auxiliary function symbols, one of which is running and the others are waiting.
%We prepare a new representation that fix the positions of called functions. 
%In order to realize it, constitute a rewrite rule along the fact that the actual function call is represented using the stack. 
A solution to fix this problem is to make such positions unique. 
An execution is represented as a term rooted by $\symb{env}$, and global variables are located at fixed positions (i.e., arguments of $\symb{env}$).
The last argument of $\symb{env}$ is used for execution of user-defined functions.
In the last argument, we fix positions where functions are called by using a so-called \emph{call stack}.
%In other words, prepare a new function symbol $\symb{stack}$ acts as a stack so that the function being executed is placed at the top of the stack. 
To this end, we prepare a binary function symbol $\symb{stack}: \sort{state} \times \sort{process} \funtype \sort{process}$ and a constant $\bot: \sort{process}$ (the empty stack).
To adapt to stacks, we change the sort of $\symb{env}$.
For example, we give $\INT \times \sort{process} \funtype \sort{env}$ to $\symb{env}$, and the initial term for the execution of the program is the following one:
\[
\symb{env}(\symb{0},\symb{stack}(\symb{main}(),\bot))
\]
In this approach, the environment has a stack $\var{s}$ to execute functions by means of the form $\symb{env}(\ldots,\var{s})$.
In calling a function $\symb{f}$ as $\symb{f}(\vec{t})$, we push $\symb{f}(\vec{t})$ as a frame for the function call to the stack $\var{s}$, and after the execution (successfully) halts, we pop the frame of the form $\symb{return}(\ldots)$ from the stack.

Along the idea above, the statements of calling functions in $\PrgrmA$ in Figure~\ref{fig:sum}---the rules of $\cRsumrec$ related to $\SumF$ or $\MainB$---are transformed into the following rules:
\[
 \left\{
 \begin{array}{r@{\>}c@{\>}l}
 	\Rule{\symb{stack}(\SumE(x,z),s) &}{& \symb{stack}(\symb{sum}(x-\symb{1}),\symb{stack}(\SumF(x,z),s))}, \\
 	\Rule{\symb{stack}(\symb{return}(y),\symb{stack}(\SumF(x,z),s)) &}{& \symb{stack}(\SumG(x,y),s)}, \\[5pt]
 	\Rule{\symb{stack}(\MainA(n),s) &}{& \symb{stack}(\symb{sum}(n),\symb{stack}(\MainB(n),s))}, \\
 	\Rule{\symb{stack}(\symb{return}(y),\symb{stack}(\MainB(n))) &}{& \symb{stack}(\SumC(n),s)} \\
 \end{array}
 \right\}
\]
The first and third rules push frames to the stack, and the second and fourth pop frames.
For a term $\symb{env}(x_1,\ldots,x_k,\symb{stack}(\ldots))$, the reduction of user-defined functions is performed at the position $k+1$ of the term, where $x_1,\ldots,x_k$ are global variables. 
For this reason, statements accessing global variables can be represented by the following form:
\[
\CRule{\symb{env}(x_1,\ldots,x_k,\symb{stack}(\symb{f}(\ldots),s))}{\symb{env}(t_1,\ldots,t_k,\symb{stack}(\symb{g}(\ldots),s))}{\varphi}
\]
Note that $s$ in the above rule is a variable.
The statement ``\verb|num = num + 1|'' in $\PrgrmA$---the rule of $\cRsumrec$ to increment $\var{num}$---is transformed into the following rule:
\[
 	\Rule{\symb{env}(\var{num},\symb{stack}(\SumA(x,z),s))}{\symb{env}(\var{num}+\symb{1},\symb{stack}(\SumB(x,z),s))}
\]
In summary, $\PrgrmA$ is transformed into the following LCTRS with the sort set $\{\INT,\BOOL,\sort{state},\sort{env},\sort{process}\}$ and the standard integer signature $\Sigmaint$:
\[
	\cRsumrecstack =
 \left\{
 \begin{array}{r@{\>}c@{\>}lr@{\,}c@{\,}l}
 	\Rule{\symb{sum}(x) &}{& \SumA(x,\symb{0})}, \\
 	\Rule{\symb{env}(\var{num},\symb{stack}(\SumA(x,z),s)) &}{& \symb{env}(\var{num}+\symb{1},\symb{stack}(\SumB(x,z),s))}, \\
 	\CRule{\SumB(x,z) &}{& \SumC(x,z) &}{& x \leq \symb{0} &}, \\
 	\CRule{\SumB(x,z) &}{& \SumE(x,z) &}{& \neg (x \leq \symb{0}) &}, \\
 	\Rule{\SumC(x,z) &}{& \SumD(x,\symb{0})}, \\
 	\Rule{\SumD(x,z) &}{& \SumI(x,z)}, \\
 	\Rule{\symb{stack}(\SumE(x,z),s) &}{& \symb{stack}(\symb{sum}(x-\symb{1}),\symb{stack}(\SumF(x,z),s))}, \\
 	\Rule{\symb{stack}(\symb{return}(y),\symb{stack}(\SumF(x,z),s)) &}{& \symb{stack}(\SumG(x,y),s)}, \\
 	\Rule{\SumG(x,z) &}{& \SumH(x,x+z)}, \\
 	\Rule{\SumH(x,z) &}{& \SumI(x,z)}, \\
 	\Rule{\SumI(x,z) &}{& \symb{return}(z)}, \\[5pt]

 	\Rule{\symb{main}() &}{& \MainA(\symb{3})}, \\
 	\Rule{\symb{stack}(\MainA(z),s) &}{& \symb{stack}(\symb{sum}(z),\symb{stack}(\MainB(z),s))}, \\
 	\Rule{\symb{stack}(\symb{return}(y),\symb{stack}(\MainB(z))) &}{& \symb{stack}(\MainC(y),s)}, \\
 	\Rule{\MainC(z) &}{& \symb{return}(\symb{0})} \\
 \end{array}
 \right\}
\]
%Since the position of a called function is fixed in the first argument of $\symb{stack}$, there is only one rule to represent the process of accessing global variables in the execution of $\symb{sum}$.
For the execution of the program, we have the reduction of $\cRsumrecstack$ illustrated in Figure~\ref{fig:sumrecstack-reduction}.

\begin{figure}[t]
\[
\begin{array}{@{}l@{\>\>\>}l@{}}
	\symb{env}(\symb{0},\symb{stack}(\symb{main}(),\bot))
	& \to_{\cRsumrecstack}
		\symb{env}(\symb{0},\symb{stack}(\MainA(\symb{3}),\bot)) \\
	& \to_{\cRsumrecstack}
		\symb{env}(\symb{0},\symb{stack}(\symb{sum}(\symb{3}),\symb{stack}(\MainB(\symb{3}),\bot))) \\
	& \to_{\cRsumrecstack}
		\symb{env}(\symb{0},\symb{stack}(\SumA(\symb{3},\symb{0}),\symb{stack}(\MainB(\symb{3}),\bot))) \\
	& \to_{\cRsumrecstack}
		\symb{env}(\symb{0}+\symb{1},\symb{stack}(\SumB(\symb{3},\symb{0}),\symb{stack}(\MainB(\symb{3}),\bot))) \\
	& \to_{\cRsumrecstack}
		\symb{env}(\symb{1},\symb{stack}(\SumB(\symb{3},\symb{0}),\symb{stack}(\MainB(\symb{3}),\bot))) \\
	& \to_{\cRsumrecstack}
		\symb{env}(\symb{1},\symb{stack}(\SumE(\symb{3},\symb{0}),\symb{stack}(\MainB(\symb{3}),\bot))) \\
	& \to_{\cRsumrecstack}
		\symb{env}(\symb{1},\symb{stack}(\symb{sum}(\symb{3}-\symb{1}),\symb{stack}(\SumF(\symb{3},\symb{0}),\symb{stack}(\MainB(\symb{3}),\bot)))) \\
	& \to_{\cRsumrecstack}
		\symb{env}(\symb{1},\symb{stack}(\symb{sum}(\symb{2}),\symb{stack}(\SumF(\symb{3},\symb{0}),\symb{stack}(\MainB(\symb{3}),\bot)))) \\

%%%%%%%%%%%%%%%%%%%%%
	& \to_{\cRsumrecstack} \cdots \\
	& \to_{\cRsumrecstack}
		\symb{env}(\symb{4},\symb{stack}(\SumG(\symb{3},\symb{3}),\symb{stack}(\MainB(\symb{3}),\bot))) \\
	& \to_{\cRsumrecstack}
		\symb{env}(\symb{4},\symb{stack}(\SumH(\symb{3},\symb{3}+\symb{3}),\symb{stack}(\MainB(\symb{3}),\bot))) \\
	& \to_{\cRsumrecstack}
		\symb{env}(\symb{4},\symb{stack}(\SumH(\symb{3},\symb{6}),\symb{stack}(\MainB(\symb{3}),\bot))) \\
	& \to_{\cRsumrecstack}
		\symb{env}(\symb{4},\symb{stack}(\SumI(\symb{3},\symb{6}),\symb{stack}(\MainB(\symb{3}),\bot))) \\
	& \to_{\cRsumrecstack}
		\symb{env}(\symb{4},\symb{stack}(\symb{return}(\symb{6}),\symb{stack}(\MainB(\symb{3}),\bot))) \\
	& \to_{\cRsumrecstack}
		\symb{env}(\symb{4},\symb{stack}(\MainC(\symb{6},\bot))) \\
	& \to_{\cRsumrecstack}
		\symb{env}(\symb{4},\symb{stack}(\symb{return}(\symb{0}))) \\
\end{array}
\]
%\vspace{-10pt}
\caption{the reduction of $\cRsumrecstack$ for the execution of the program for \texttt{sum}.}
\label{fig:sumrecstack-reduction}
\end{figure}

The function symbol $\symb{stack}$ is a defined symbol of $\cRsumrecstack$, while it looks a constructor for stacks.
If we would like the resulting LCTRS to be a constructor system, rules performing ``push'' and ``pop'' for stacks may be generated as rules for $\symb{env}$.
More precisely, we generate
$
\Rule{\symb{env}(\Xvec,\symb{stack}(t,\var{s}))}{\symb{env}(\Xvec,\symb{stack}(t',s'))}
$
instead of $\Rule{\symb{stack}(t,\var{s})}{\symb{stack}(t',s')}$.
%and instead of $\Rule{\symb{stack}(\symb{return}(z),\symb{stack}(\symb{f'}(\Yvec),\var{s}))}{\symb{stack}(\symb{f''}(\Yvec),\var{s})}$, we generate
%\[
%\Rule{\symb{env}(\Xvec,\symb{stack}(\symb{return}(z),\symb{stack}(\symb{f'}(\Yvec),\var{s})))}{\symb{env}(\Xvec,\symb{stack}(\symb{f''}(\Yvec),\var{s}))}.
%\]
%For example, $\cRsumrecstack$ can be modified to the following constructor LCTRS:
%\[
%%	\cRsumrecstack' =
% \left\{
% \begin{array}{r@{\>}c@{\>}l@{\hspace{-30pt}}r@{\,}c@{\,}l}
% 	\Rule{\symb{env}(\var{num},\symb{stack}(\symb{sum}(x),s)) &}{& \symb{env}(\var{num}+\symb{1},\symb{stack}(\SumA(x,\symb{0}),s))} \\
% 	\CRule{\SumA(x,z) &}{& \SumC(x,\symb{0}) &}{& x \leq \symb{0} &} \\
% 	\CRule{\symb{env}(\var{num},\symb{stack}(\SumA(x,z),s) &}{& \symb{env}(\var{num},\symb{stack}(\symb{sum}(x-\symb{1}),\symb{stack}(\SumB(x,z),s)))) &}{& \neg (x \leq \symb{0}) &} \\
% 	\Rule{\symb{env}(\var{num},\symb{stack}(\symb{return}(y),\symb{stack}(\SumB(x,z),s))) &}{& \symb{env}(\var{num},\symb{stack}(\SumC(x,x + y),s))} \\[5pt]
% 	\Rule{\symb{main}() &}{& \MainA(\symb{3})} \\
% 	\Rule{\symb{env}(\var{num},\symb{stack}(\MainA(n),s)) &}{& \symb{env}(\var{num},\symb{stack}(\symb{sum}(n),\symb{stack}(\MainB(n),s)))} \\
% 	\Rule{\symb{env}(\var{num},\symb{stack}(\symb{return}(y),\symb{stack}(\MainB(n)))) &}{& \symb{env}(\var{num},\symb{stack}(\symb{return}(\symb{0}),s))} \\
% \end{array}
% \right\}
%\]
%The reduction of this LCTRS from $\symb{env}(\symb{0},\symb{stack}(\symb{main}(),\bot))$ is the same as that in Figure~\ref{fig:sumrecstack-reduction}.

\section{Formalizing the Transformation Using Stacks}
\label{sec:formalization}

In this section, we formalize the idea of using call stacks, which is illustrated in Section~\ref{sec:new-transformation}, showing a precise transformation of {\SIMP} programs into LCTRSs.

%To simplify the discussion, we assume that a given {\SIMP} program has a single function-declaration.
In the following, we deal with a {\SIMP} program $P$ which is of the following form:
\begin{equation}
\label{eqn:P-form}
\begin{array}{@{}l@{}}
\Assign{\IntVarDecl{x_1}}{n_1};~\ldots;~\Assign{\IntVarDecl{x_k}}{n_k};\\
\IntFunDecl{\symb{f}_1}{\IntVarDecl{y_{1,1}},\ldots,\IntVarDecl{y_{1,m_1}}}~\{~\alpha_1~~\Return{e_1};~\}\\
~~~~~\ldots \\
\IntFunDecl{\symb{f}_{k'}}{\IntVarDecl{y_{k',1}},\ldots,\IntVarDecl{y_{k',m_{k'}}}}~\{~\alpha_{k'}~~\Return{e_{k'}};~\}\\
\end{array}
\end{equation}
where $\alpha_1,\ldots,\alpha_{k'}$ are statements with local-variable declarations and  no function other than $\symb{f}_1,\ldots,\symb{f}_{k'}$ is called in $\alpha_1,\ldots,\alpha_{k'}$.
Note that $\symb{f}_1,\ldots,\symb{f}_{k'}$ may be self- or mutually recursive.
We abuse integer and boolean expressions of {\SIMP} programs as theory terms and formulas, respectively, over the standard integer signature $\Sigmaint$.
In the following, we denote the sequences $x_1,\ldots,x_k$ and $y_{i,1},\ldots,y_{i,m_i}$ by $\Xvec$ and  $\Yvec[i]$, respectively, and the notation $\Yvec$ stands for $\Yvec[i]$ for some $i \in \{1,\ldots,k'\}$.

First, we define an auxiliary function $\Trans$ that takes a term $t$, a statement $\beta$ with variable declarations, and a non-negative integer $i$ as input, and returns a triple $(u,\cR_\beta,j)$ of a term $u$, a set $\cR_\beta$ of constrained rewrite rules, and a non-negative integer $j$.
The resulting rewrite rules in $\cR_\beta$ reduce an instance of $\symb{env}(\Xvec,\symb{stack}(t,s))$ to an instance of $\symb{env}(\Xvec,\symb{stack}(u,s))$:
if the instance of $\symb{env}(\Xvec,\symb{stack}(t,s))$ corresponds to a configuration $\Config{\beta}{\sigma}{\sigma'}$,
then the instance of $\symb{env}(\Xvec,\symb{stack}(u,s))$ corresponds to a configuration $\Config{\Skip}{\sigma''}{\sigma'''}$ such that 
$\Config{\beta}{\sigma}{\sigma'} \Downarrow_P \Config{\Skip}{\sigma''}{\sigma'''}$.
The input term $t$ is of the form either $\symb{f}_{k''}(\Yvec[k''])$ or $\symb{u}_{i'}(\Yvec[k''],z_{k'',1},\ldots,z_{k'',m'_{k''}})$
where $z_{k'',1},\ldots,z_{k'',m'_{k''}}$ are locally declared variables in $\alpha_{k''}$ and $\symb{u}_{i'}$ is a newly introduced function symbol with $i' < i < j$.
The resulting term $u$ is of the form of $\symb{u}_{j'}(\Yvec[k''],z_{k'',1},\ldots,z_{k'',m''_{k''}})$
where $m'_{k''} \leq m''_{k''}$, $z_{k'',1},\ldots,z_{k'',m''_{k''}}$ are locally declared variables in $\alpha_{k''}$, and $\symb{u}_{j'}$ is a newly introduced function symbol with $i \leq j' < j$.
In the following, we denote the sequence $z_{k'',1},\ldots,z_{k'',m'_{k''}}$ by $\Zvec[k'']$, and the sequence $e'_1,\ldots,e'_{m_i}$ of integer expressions by $\Evec[i]$, and the notation $\Zvec$ stands for $\Zvec[k'']$ for some $k'' \in \{1,\ldots,k'\}$.
\begin{definition}
\label{def:subconverter}
The auxiliary function $\Trans$ is defined as follows:
\begin{itemize}
	\item $\Trans(t, ~ \epsilon, ~ i) = (t, \emptyset, i)$,
%	\item $\Trans(t, ~ \alpha_1;\alpha_2, ~ i) = (u_2, \cR_1 \cup \cR_2, j_2)$,
%		where 
%		\begin{itemize}
%			\item $\Trans(t, \alpha_1, i) = (u_1, \cR_1, j_1)$,
%			and
%			\item $\Trans(u_1, \alpha_2, j_1) = (u_2, \cR_2, j_2)$,
%		\end{itemize}
	\item $\Trans(g(\Yvec,\Zvec), ~ \Assign{\IntVarDecl{z'}}{n};~\beta, ~ i) = 
		(u, \{~\Rule{g(\Yvec,\Zvec)}{\symb{u}_i(\Yvec,\Zvec,n)}~\}\cup\cR_\beta, j)
		$, where 
		\begin{itemize}
			\item $\Trans(\symb{u}_i(\Yvec,\Zvec,z'), \beta, i+1) = (u, \cR_\beta, j)$,
		\end{itemize}
	\item $\Trans(g(\Yvec,\Zvec), ~ \Assign{z'}{e};~\beta, ~ i) = 
		(u, \{~ \Rule{C[g(\Yvec,\Zvec)]}{(C[\symb{u}_i(\Yvec,\Zvec)])\{z'\mapsto e\}} ~\} \cup \cR_\beta, j)
		$
	if $e$ is an integer expression,
	where 
	\begin{itemize}
		\item if $\{\Xvec\} \cap (\{z'\} \cup \FVar(e)) \ne \emptyset$ then $C[\,] = \symb{env}(\Xvec,\symb{stack}(\Hole,w))$ with a fresh variable $w \notin \{\Xvec,\Yvec,\Zvec\}$, and otherwise $C[\,] = \Hole$, and
		\item $\Trans(\symb{u}_i(\Yvec,\Zvec,z'), \beta, i+1) = (u, \cR_\beta, j)$,
	\end{itemize}

	\item $\Trans(g(\Yvec,\Zvec), ~ \Assign{z'}{\symb{f}_{k''}(\Evec[k''])};~\beta, ~ i) =$
	\[\hspace{-.8ex}
	 (u,\!
	 \left\{
	 \begin{array}{@{}r@{\,}c@{\,}l@{\!}}
	 \Rule{C[\symb{stack}(g(\Yvec,\Zvec),w)] &}{& C[\symb{stack}(\symb{f}_{k''}(\Evec[k'']),\symb{stack}(\symb{u}_i(\Yvec,\Zvec),w))]}, \\
	 \Rule{C'[\symb{stack}(\symb{return}(z''),\symb{stack}(\symb{u}_i(\Yvec,\Zvec),w))] &}{& (C'[\symb{stack}(\symb{u}_{i+1}(\Yvec,\Zvec),w)])\{z' \mapsto z''\}} \\
	 \end{array}
	 \right\} \!\cup \cR_\beta, j)
	\]
	where 
	\begin{itemize}
		\item $w,z''$ are different fresh variables not in $\{\Xvec,\Yvec,\Zvec\}$,
		\item if $\{\Xvec\} \cap \FVar(\Evec[k'']) \ne \emptyset$ then $C[\,] = \symb{env}(\Xvec,\Hole)$, and otherwise $C[\,] = \Hole$,
		\item if $z' \in \{\Xvec\}$ then $C'[\,] = \symb{env}(\Xvec,\Hole)$, and otherwise $C'[\,] = \Hole$,
		and
		\item $\Trans(\symb{u}_{i+1}(\Yvec,\Zvec), \beta, i+2) = (u, \cR_\beta, j)$,
	\end{itemize}

	\item $\Trans(g(\Yvec,\Zvec), ~ \If{\varphi}{\beta_1}{\beta_2}~\beta, ~ i) =$
	\[
	 (u, 
	 \left\{
	 \begin{array}{r@{\>}c@{\>}l@{~~}c@{\,}c@{\,}c@{~~~~~~}r@{\>}c@{\>}l}
	 \CRule{C[g(\Yvec,\Zvec)] &}{& C[\symb{u}_i(\Yvec,\Zvec)] &}{& \varphi &}, &
	 \Rule{u_1 &}{& \symb{u}_{j_2}(\Yvec,\Zvec)}, \\
	 \CRule{C[g(\Yvec,\Zvec)] &}{& C[\symb{u}_{j_1+1}(\Yvec,\Zvec)] &}{& \neg \varphi &}, &
	 \Rule{u_2 &}{& \symb{u}_{j_2}(\Yvec,\Zvec)} \\
	 \end{array}
	 \right\} \cup \cR_{\beta_1} \cup \cR_{\beta_2} \cup \cR_\beta, j)
	\]
	where 
	\begin{itemize}
		\item $\Trans(\symb{u}_i(\Yvec,\Zvec),\beta_1,i+1)=(u_1,\cR_{\beta_1},j_1)$,
		\item $\Trans(\symb{u}_{j_1+1}(\Yvec,\Zvec),\beta_2,j_1+1)=(u_2,\cR_{\beta_2},j_2)$,
		\item if $\{\Xvec\} \cap \FVar(\varphi) \ne \emptyset$ then $C[\,] = \symb{env}(\Xvec,\symb{stack}(\Hole,w))$ with a fresh variable $w \notin \{\Xvec,\Yvec,\Zvec\}$, and otherwise $C[\,] = \Hole$,
		and
		\item $\Trans(\symb{u}_{j_2}(\Yvec,\Zvec),\beta,j_2+1)=(u,\cR_\beta,j)$,
	\end{itemize}

	\item $\Trans(g(\Yvec,\Zvec), ~ \While{\varphi}{\alpha}~\beta, ~ i) =$
	\[
	 (u', 
	 \left\{
	 \begin{array}{r@{\>}c@{\>}l@{~~}c@{\,}c@{\,}c@{~~~~~~}r@{\>}c@{\>}l}
	 \CRule{C[g(\Yvec,\Zvec)] &}{& C[\symb{u}_i(\Yvec,\Zvec)] &}{& \varphi &}, &
	 \Rule{u &}{& g(\Yvec,\Zvec)},
	 \\
	 \CRule{C[g(\Yvec,\Zvec)] &}{& C[\symb{u}_j(\Yvec,\Zvec)] &}{& \neg \varphi &} \\
	 \end{array}
	 \right\} \cup \cR_\alpha \cup \cR_\beta, j')
	\]
	where 
	\begin{itemize}
		\item $\Trans(\symb{u}_i(\Yvec,\Zvec),\alpha,i+1)=(u,\cR_\alpha,j)$,
		\item if $\{\Xvec\} \cap \FVar(\varphi) \ne \emptyset$ then $C[\,] = \symb{env}(\Xvec,\symb{stack}(\Hole,w))$ with a fresh variable $w \notin \{\Xvec,\Yvec,\Zvec\}$, and otherwise $C[\,] = \Hole$.
		and
		\item $\Trans(\symb{u}_j(\Yvec,\Zvec),\beta,j+1)=(u',\cR_\beta,j')$,
	\end{itemize}
\end{itemize}
The sorts of generated symbols are determined as follows:
$\symb{f}_1,\ldots,\symb{f}_{k'},\symb{u}_i,\symb{u}_{i+1},\ldots: \INT \times \cdots \times \INT \funtype \sort{state}$,
$\symb{return}: \INT  \funtype \sort{state}$,
$\symb{env}: \INT \times \cdots \times \INT \times \sort{process} \funtype \sort{env}$,
$\symb{stack}: \sort{state} \times \sort{process} \funtype \sort{process}$,
and 
$\bot: \sort{process}$.
\end{definition}
Using $\Trans$, the transformation illustrated in Section~\ref{sec:new-transformation} is defined as follows.
\begin{definition}
\label{def:converter}
%For a non-negative integer $i$, 
We define $\Convert$ %for $P$
by
$
\Convert(P) = \bigcup_{i=1}^{k'} (\cR_i \cup \{~\Rule{C_i[\var{u}_i]}{C_i[\symb{return}(e_i)]}~\})
$,
where $j_1=1$\,%
	\footnote{
	The third argument of $\Trans$ is used to generate new function symbols of the form $\symb{u}_i$.
	We do not have to start with $1$, and we can put any non-negative integer into the third argument of $\Trans$ in order to, e.g., avoid the introduction of the same function symbol for two different inputs.
	}
 and for each $i \in \{1,\ldots,k'\}$,
\begin{itemize}
	\item $\Trans(\symb{f}_i(\Yvec[i]), \alpha_i, j_i)= (u_i, \cR_i, j_{i+1})$,	and
	\item if $\{\Xvec\} \cap \FVar(e_i) \ne \emptyset$ then $C_i[\,] = \symb{env}(\Xvec,\symb{stack}(\Hole,w))$ with a fresh variable $w \notin \{\Xvec\}\cup\FVar(u_i)$, and otherwise $C_i[\,] = \Hole$.
\end{itemize}
\end{definition}
By definition, it is clear that $\Convert(P)$ is an LCTRS with the sort set $\{\INT,\BOOL,\sort{state},\sort{env},\sort{process}\}$ and the standard integer signature $\Sigmaint$.
Note that Definitions~\ref{def:subconverter} and~\ref{def:converter} follow the formulation in~\cite{FNSKS08b}.
Note also that $\cR$ is orthogonal, any term reachable from $(\symb{env}(\Xvec,\symb{stack}(\symb{f}_i(\Yvec[i]),s)))(\sigma_0\cup\sigma_1)$ with a normal form $s$ has at most one redex that is not for $\cRcalc$.%
\footnote{
More precisely, the redex of a term reachable from $(\symb{env}(\Xvec,\symb{stack}(\symb{f}(\Yvec),s)))(\sigma_0\cup\sigma_1)$ is at the root position, position $k+1$, position $(k+1).1$, or position $(k+1).1.p$ for some $p$.
}
Since the reduction of $\cRcalc$ is convergent, we restrict the reduction of $\cRcalc$ to the \emph{leftmost} one.
Then, any subderivation $t \mathrel{\to^*_\cR} t'$ of a derivation from $(\symb{env}(\Xvec,\symb{stack}(\symb{f}_i(\Yvec[i]),s)))(\sigma_0\cup\sigma_1)$ has at most one pass from $t$ to $t'$.

%The third component $j$ of the resulting triple $(u, \cR, j)$ for $\Trans(\symb{f}(\Yvec), \alpha, 1)$ is not used for $\Convert(P)$.
%However, when considering the case where $P$ has two or more function declarations, we pass $j$ to the call of $\Trans$ for other functions.
%For example, if $P$ has two function declarations 
%\[
%\IntFunDecl{\symb{f}}{\IntVarDecl{y_1},\ldots,\IntVarDecl{y_m}}~=~\{~\alpha~~\Return{e};~\}
%\quad
%\mbox{and}
%\quad
%\IntFunDecl{\symb{f}'}{\IntVarDecl{y_1},\ldots,\IntVarDecl{y_{m'}}}~\{~\alpha'~~\Return{e'};~\}
%\]
%then we first compute $\Trans(\symb{f}(\Yvec), \alpha, 1)$ to get $(u, \cR, j)$, and then compute $\Trans(\symb{f}'(y_1,\ldots,y_{m'}), \alpha', j)$ to get $(u',\cR',j')$, generating $\cR \cup \cR' \cup \{\, \Rule{C[\var{u}]}{C[\symb{return}(e)]},~\Rule{C'[\var{u'}]}{C'[\symb{return}(e')]} \,\}$ as a result of $\Convert(P)$.

\begin{example}
Consider the program $\PrgrmA$ in Figure~\ref{fig:sum}.
We have that $\Convert(\PrgrmA)=\cRsumrecstack$.
\end{example}

Finally, we show correctness of the transformation $\Convert$.
Recall that $P$ is assumed to be of the form~(\ref{eqn:P-form}).
%$
%\Assign{\IntVarDecl{x_1}}{n_1};~\ldots;~\Assign{\IntVarDecl{x_k}}{n_k};~\IntFunDecl{\symb{f}_1}{\IntVarDecl{y_{1,1}},\ldots,\IntVarDecl{y_{1,m_1}}}~\{~\alpha_1~~\Return{e_1};~\}~\ldots~\IntFunDecl{\symb{f}_{k'}}{\IntVarDecl{y_{k',1}},\ldots,\IntVarDecl{y_{k',m_{k'}}}}~\{~\alpha_{k'}~~\Return{e_{k'}};~\}
%$.
We first show two auxiliary lemmas.
\begin{lemma}
\label{lem:calc}
Let $\cR$ be an LCTRS, $e$ an integer expression, $n$ an integer, and $\sigma$ an assignments for $\FVar(e)$.
Then, 
$(e,\sigma) \toExpr n$
if and only if 
$e\sigma \mathrel{\to^*_{\cR}} n$.
\end{lemma}
\begin{proof}
Trivial by the definitions of $\toExpr$ and $\cRcalc$.	
\qed
\end{proof}
%To prove correctness of $\Trans$ by induction, we introduce special height of terms and derivations as follows:
%$\Height{n}=\Height{x}=\Height{\bot}=0$ for $n\in \Int$ and $x \in \cV$;
%$\Height{g(t_1,\ldots,t_n)}=0$ for $g \in \{\symb{f},+,-\}\cup\{\symb{u}_j \mid \mbox{$\symb{u}_j$ is introduced by $\Trans$}\}$;
%$\Height{\symb{env}(\Xvec,t)}=\Height{t}$;
%$\Height{\symb{stack}(t,s)}=1+\Height{s}$;
%$\Height{t_1 \mathrel{\to_\cR} t_2 \mathrel{\to_\cR} \cdots \mathrel{\to_\cR} t_n}=\max(\Height{t_1},\ldots,\Height{t_n})$.
%Note that the height of terms coincides with the depth of nesting function calls.
%Note also that for any subderivation $t \mathrel{\to^*_\cR} t'$ of derivations starting from $(\symb{env}(\Xvec,\symb{stack}(\symb{f}(\Yvec),s)))(\sigma_0\cup\sigma_1)$ with a normal form $s$, the height $\Height{t \mathrel{\to^*_\cR} t'}$ is unique because the subderivation $t \mathrel{\to^*_\cR} t'$ has at most one pass from $t$ to $t'$.
%
%Let $t,t'$ be terms, $p,p'$ positions of $t$, and the reduction $t \mathrel{\to_\cR} t'$ occur at $p$.
%We write $t \mathrel{\to_{\leq p',\cR}} t'$ if $p \leq p'$.

\begin{lemma}[Correctness of $\Trans$]
\label{lem:correctness}
Let %$i$ be a non-negative integer, 
$\cR = \Convert(P)$, and $\beta$ a substatement of $\alpha_i$ for some $i\in\{1,\ldots,k'\}$ (i.e., $\beta$ appears in $\alpha_i$).
%$n'_1,\ldots,n'_m \in \Int$,
%$\sigma_0=\Updates{\emptyset}{x_1\mapsto n_1,\ldots,x_k\mapsto n_k}$,
%$\sigma_1=\Updates{\emptyset}{y_1\mapsto n'_1,\ldots,y_m\mapsto n'_m}$,
%and
%$t_0=(\symb{env}(\Xvec,\symb{stack}(\symb{f}(\Yvec),s)))(\sigma_0\cup\sigma_1)$ with a normal form $s$.
Then, both of the following hold:
\begin{enumerate}
	\renewcommand{\labelenumi}{(\alph{enumi})}
	\item
		$\Trans(t,\beta,i')$ for any $t$ and $i'$ is defined,
		and
	\item
	$\Trans(t,\beta,i')$ for some $t$ and $i'$ is computed during the computation of $\Convert(P)$.
\end{enumerate}
Suppose that $\Trans(g(\Yvec,\Zvec),\beta,i')$ is computed for $\Convert(P)$.
Let $\Trans(g(\Yvec,\Zvec),\beta,i')=(u,\cR_\beta,j)$, and $s$ be a normal form of $\cR$, $\sigma_0,\sigma_0'$ assignments for $\GVar(P)$, and $\sigma_1,\sigma_1'$ assignments for $\{\Yvec,\Zvec\}\cup(\FVar(\beta)\setminus\{\Xvec\})$.
Then, both of the following hold:
\begin{enumerate}
	\setcounter{enumi}{2}
	\renewcommand{\labelenumi}{(\alph{enumi})}
	\item 
		$\cR_\beta \subseteq \cR$,
	\item 
		$\Config{\beta}{\sigma_0}{\sigma_1} \Downarrow_P \Config{\Skip}{\sigma_0'}{\sigma_1'}$ if and only if 
		\[
		(\symb{env}(\Xvec,\symb{stack}(g(\Yvec,\Zvec),s))(\sigma_0\cup\sigma_1) \mathrel{\to^*_{\cR}} (\symb{env}(\Xvec,\symb{stack}(u,s)))(\sigma_0'\cup\sigma_1').
		\]
\end{enumerate}
\end{lemma}
\begin{proof}
By definition, it is clear that (a)--(c) hold.
Using Lemma~\ref{lem:calc}, the \textit{only-if} and \textit{if} parts of (d) can be proved by induction on the height of the inference for $\Config{\beta}{\sigma_0}{\sigma_1} \Downarrow_P \Config{\Skip}{\sigma_0'}{\sigma_1'}$ and the length of $\to^*_{\cR}$-steps, respectively.
The difference from the proof in~\cite{FNSKS08b} is the treatment of global variables and function calls, while \cite{FNSKS08b} adopts a small-step semantics for their imperative programs.
Below, we only show the case where $\beta$ is $\Assign{z'}{\FCall{\symb{f}_{k''}}{\Evec[k'']}};~\beta'$ for some $k'' \in \{1,\ldots,k'\}$.
Let $\Trans(g(\Yvec,\Zvec),\beta,i')$ return
\[
	 (u,\!
	 \left\{
	 \begin{array}{@{}r@{\,}c@{\,}l@{}}
	 \Rule{C'[\symb{stack}(g(\Yvec,\Zvec),w)] &}{& C'[\symb{stack}(\symb{f}_{k''}(\Evec[k'']),\symb{stack}(\symb{u}_{i'}(\Yvec,\Zvec),w))]}, \\
	 \Rule{C''[\symb{stack}(\symb{return}(z''),\symb{stack}(\symb{u}_{i'}(\Yvec,\Zvec),w))] &}{& (C''[\symb{stack}(\symb{u}_{i'+1}(\Yvec,\Zvec),w)])\{z' \mapsto z''\}} \\
	 \end{array}
	 \right\} \! \cup \cR_{\beta'}, j)
	\]
	where 
	\begin{itemize}
		\item $w,z''$ are different fresh variables not in $\{\Xvec,\Yvec,\Zvec\}$,
		\item if $\{\Xvec\} \cap \FVar(\Evec[k'']) \ne \emptyset$ then $C'[\,] = \symb{env}(\Xvec,\Hole)$, and otherwise $C'[\,] = \Hole$,
		\item if $z' \in \{\Xvec\}$ then $C''[\,] = \symb{env}(\Xvec,\Hole)$, and otherwise $C''[\,] = \Hole$,
		and
		\item $\Trans(\symb{u}_{i'+1}(\Yvec,\Zvec),\beta',i'+2)=(u,\cR_{\beta'},j)$,
	\end{itemize}
Then, it follows from (c) that the above two rules are included in $\cR$.

We first show the \textit{only-if} part.
Assume that
$
\Config{\Assign{z'}{\FCall{\symb{f}_{k''}}{\Evec[k'']}};~\beta'}{\sigma_0}{\sigma_1} \Downarrow_P \Config{\Skip}{\sigma_0'}{\sigma_1'}
$ holds
with
\begin{itemize}
	\item $(e'_i,\sigma_0\cup\sigma_1) \toExpr n_i$ for all $1 \leq i \leq m_{k''}$,
	\item $\Config{\alpha_{k''}}{\sigma_0}{\sigma_1} \Downarrow_P \Config{\Skip}{\sigma_0''}{\sigma_1''}$,
	\item $(e_{k''},\sigma_0''\cup\sigma_1'') \toExpr n$,
	and
	\item $\Config{\beta'}{\sigma_0'''}{\sigma_1'''} \Downarrow_P \Config{\Skip}{\sigma_0'}{\sigma_1'}$
\end{itemize}
where 
	\begin{itemize}
		\item if $z' \in \GVar(P)$ then $\sigma_0'''=\Update{\sigma_0''}{z'}{n}$, and otherwise $\sigma_0'''=\sigma_0''$,
			and
		\item if $z' \in \GVar(P)$ then $\sigma_1'''=\sigma_1$, and otherwise $\sigma_1'''=\Update{\sigma_1}{z'}{n}$.
	\end{itemize}
It follows from Lemma~\ref{lem:calc} and $\Rule{C'[\symb{stack}(g(\Yvec,\Zvec),w)]}{C'[\symb{stack}(\symb{f}_{k''}(\Evec[k'']),\symb{stack}(\symb{u}_{i'}(\Yvec,\Zvec),w))]} \in \cR$ that
\[
\begin{array}{@{}l@{\>}c@{\>}l@{}}
\lefteqn{(\symb{env}(\Xvec,\symb{stack}(g(\Yvec,\Zvec),s)))(\sigma_0\cup\sigma_1)} \\
~~~~
&
\mathrel{\to_\cR}
&
(\symb{env}(\Xvec,\symb{stack}(\symb{f}_{k''}(\Evec[k'']),\symb{stack}(\symb{u}_{i'}(\Yvec,\Zvec),s))))(\sigma_0\cup\sigma_1)
\\
&
=
&
(\symb{env}((\Xvec)\sigma_0,\symb{stack}(\symb{f}_{k''}((\Evec[k''])(\sigma_0\cup\sigma_1)),\symb{stack}(\symb{u}_{i'}((\Yvec)\sigma_1,(\Zvec)\sigma_1),s))))
\\
&
\mathrel{\to^*_\cR}
&
(\symb{env}((\Xvec)\sigma_0,\symb{stack}(\symb{f}_{k''}(n_1,\ldots,n_{m_{k''}}),\symb{stack}(\symb{u}_{i'}((\Yvec)\sigma_1,(\Zvec)\sigma_1),s)))).
\\
\end{array}
\]
By definition, $\Trans(\symb{f}_{k''}(\Yvec[k'']),\alpha_{k''},j_{k''})$ is computed, and let $\Trans(\symb{f}_{k''}(\Yvec[k'']),\alpha_{k''},j_{k''})=(u_{k''},\cR_{\alpha_{k''}},j_{k''+1})$.
Then, by definition, we have that $\cR_{\alpha_{k''}}\cup\{~\Rule{C_{k''}[u_{k''}]}{C_{k''}[\symb{return}(e_{k''})]}~\} \subseteq \cR$ where $C_{k''}[\,]$ is a context defined in Definition~\ref{def:converter}.
Let $\sigma_2=\{y_1\mapsto n_1,~\ldots,~y_{m_{k''}} \mapsto n_{m_{k''}}\}$.
Then, by the induction hypothesis, we have that 
\[
(\symb{env}(\Xvec,\symb{stack}(\symb{f}_{k''}(\Yvec[k'']),s'))(\sigma_0\cup\sigma_2)
\mathrel{\to^*_\cR}
(\symb{env}(\Xvec,\symb{stack}(u_{k''},s'))(\sigma_0''\cup\sigma_1'').
\]
for an arbitrary term $s'$.
Thus, we have that
\[
\begin{array}{@{}l@{\>}c@{\>}l@{}}
\lefteqn{(\symb{env}((\Xvec)\sigma_0,\symb{stack}(\symb{f}_{k''}(n_1,\ldots,n_{m_{k''}}),\symb{stack}(\symb{u}_{i'}((\Yvec)\sigma_1,(\Zvec)\sigma_1),s))))}
\\
~~~~
&
=
&
\symb{env}((\Xvec)\sigma_0,\symb{stack}(\symb{f}_{k''}((\Yvec[k''])\sigma_2),\symb{stack}(\symb{u}_{i'}((\Yvec)\sigma_1,(\Zvec)\sigma_1),s))))
\\
&
\mathrel{\to^*_\cR}
&
\symb{env}((\Xvec)\sigma_0'',\symb{stack}(u_{k''}(\sigma_0''\cup\sigma_1''),\symb{stack}(\symb{u}_{i'}((\Yvec)\sigma_1,(\Zvec)\sigma_1),s)))).
\\
\end{array}
\]
It follows from $\Rule{C_{k''}[u_{k''}]}{C_{k''}[\symb{return}(e_{k''})]} \in \cR$ and Lemma~\ref{lem:calc} that
\[
\begin{array}{@{}l@{\>}c@{\>}l@{}}
\lefteqn{\symb{env}((\Xvec)\sigma_0'',\symb{stack}(u_{k''}(\sigma_0''\cup\sigma_1''),\symb{stack}(\symb{u}_{i'}((\Yvec)\sigma_1,(\Zvec)\sigma_1),s))))}
\\
~~~~
&
\mathrel{\to^*_\cR}
&
\symb{env}((\Xvec)\sigma_0'',\symb{stack}(\symb{return}(e_{k''}(\sigma_0''\cup\sigma_1'')),\symb{stack}(\symb{u}_{i'}((\Yvec)\sigma_1,(\Zvec)\sigma_1),s))))
\\
&
\mathrel{\to^*_\cR}
&
\symb{env}((\Xvec)\sigma_0'',\symb{stack}(\symb{return}(n),\symb{stack}(\symb{u}_{i'}((\Yvec)\sigma_1,(\Zvec)\sigma_1),s)))).
\\
\end{array}
\]
Since $\Rule{C''[\symb{stack}(\symb{return}(z''),\symb{stack}(\symb{u}_{i'}(\Yvec,\Zvec),w))]}{(C''[\symb{stack}(\symb{u}_{i'+1}(\Yvec,\Zvec),w)])\{z' \mapsto z''\}} \in \cR$, we have that
\[
\begin{array}{@{}l@{}}
\symb{env}((\Xvec)\sigma_0'',\symb{stack}(\symb{return}(n),\symb{stack}(\symb{u}_{i'}((\Yvec)\sigma_1,(\Zvec)\sigma_1),s))))
\\
~~~~
{} \mathrel{\to^*_\cR}
\symb{env}((\Xvec)\sigma_0''',\symb{stack}(\symb{u}_{i'+1}((\Yvec)\sigma_1''',(\Zvec)\sigma_1'''),s)))
=
(\symb{env}(\Xvec,\symb{stack}(\symb{u}_{i'+1}(\Yvec,\Zvec),s))))(\sigma_0'''\cup\sigma_1''').
\\
\end{array}
\]
By the induction hypothesis, we have that
\[
(\symb{env}(\Xvec,\symb{stack}(\symb{u}_{i'+1}(\Yvec,\Zvec),s))))(\sigma_0'''\cup\sigma_1''')
\mathrel{\to^*_\cR}
(\symb{env}(\Xvec,\symb{stack}(u,s))))(\sigma_0'\cup\sigma_1').
\]
Therefore, the claim holds.

Next, we show the \textit{if} part.
Assume that
\[
(\symb{env}(\Xvec,\symb{stack}(g(\Yvec,\Zvec),s)))(\sigma_0\cup\sigma_1)
\mathrel{\to^*_\cR}
(\symb{env}(\Xvec,\symb{stack}(\symb{u}_{i'+1}(\Yvec,\Zvec),s)))(\sigma_0'\cup\sigma_1').
\]
Then, since derivations are unique, we can let the above derivation be the following one:
\[
\begin{array}{@{}l@{\>}c@{\>}l@{}}
\lefteqn{(\symb{env}(\Xvec,\symb{stack}(g(\Yvec,\Zvec),s)))(\sigma_0\cup\sigma_1)} \\
~~~~
&
\mathrel{\to_\cR}
&
(\symb{env}(\Xvec,\symb{stack}(\symb{f}_{k''}(\Evec[k'']),\symb{stack}(\symb{u}_{i'}(\Yvec,\Zvec),s))))(\sigma_0\cup\sigma_1)
\\
&
=
&
(\symb{env}((\Xvec)\sigma_0,\symb{stack}(\symb{f}_{k''}((\Evec[k''])(\sigma_0\cup\sigma_1)),\symb{stack}(\symb{u}_{i'}((\Yvec)\sigma_1,(\Zvec)\sigma_1),s))))
\\
&
\mathrel{\to^*_\cR}
&
(\symb{env}((\Xvec)\sigma_0,\symb{stack}(\symb{f}_{k''}(n_1,\ldots,n_{m_{k''}}),\symb{stack}(\symb{u}_{i'}((\Yvec)\sigma_1,(\Zvec)\sigma_1),s))))
\\
&
=
&
\symb{env}((\Xvec)\sigma_0,\symb{stack}(\symb{f}_{k''}((\Yvec[k''])\sigma_2),\symb{stack}(\symb{u}_{i'}((\Yvec)\sigma_1,(\Zvec)\sigma_1),s))))
\\
&
\mathrel{\to^*_\cR}
&
\symb{env}((\Xvec)\sigma_0'',\symb{stack}(u_{k''}(\sigma_0''\cup\sigma_1''),\symb{stack}(\symb{u}_{i'}((\Yvec)\sigma_1,(\Zvec)\sigma_1),s))))
\\
&
\mathrel{\to^*_\cR}
&
\symb{env}((\Xvec)\sigma_0'',\symb{stack}(\symb{return}(e_{k''}(\sigma_0''\cup\sigma_1'')),\symb{stack}(\symb{u}_{i'}((\Yvec)\sigma_1,(\Zvec)\sigma_1),s))))
\\
&
\mathrel{\to^*_\cR}
&
\symb{env}((\Xvec)\sigma_0'',\symb{stack}(\symb{return}(n),\symb{stack}(\symb{u}_{i'}((\Yvec)\sigma_1,(\Zvec)\sigma_1),s))))
\\
& =
&
\symb{env}(\Xvec,\symb{stack}(\symb{return}(n),\symb{stack}(\symb{u}_{i'}(\Yvec,\Zvec),s))))(\sigma_0''\cup\sigma_1'')
\\
&
\mathrel{\to_\cR}
&
\symb{env}(\Xvec,\symb{stack}(\symb{u}_{i'+1}(\Yvec,\Zvec),s)))(\sigma_0'''\cup\sigma_1''')
\\
&
\mathrel{\to^*_\cR}
&
(\symb{env}(\Xvec,\symb{stack}(u,s))))(\sigma_0'\cup\sigma_1')
\\
\end{array}
\]
where 
	\begin{itemize}
		\item $\sigma_2=\{y_1\mapsto n_1,~\ldots,~y_{m_{k''}} \mapsto n_{m_{k''}}\}$,
		\item if $z' \in \GVar(P)$ then $\sigma_0'''=\Update{\sigma_0''}{z'}{n}$, and otherwise $\sigma_0'''=\sigma_0''$,
			and
		\item if $z' \in \GVar(P)$ then $\sigma_1'''=\sigma_1$, and otherwise $\sigma_1'''=\Update{\sigma_1}{z'}{n}$.
	\end{itemize}
It follows from Lemma~\ref{lem:calc} and the induction hypothesis that
\begin{itemize}
	\item $(e_i,\sigma_0\cup\sigma_1) \toExpr n_i$ for all $1 \leq i \leq m$,
	\item $\Config{\alpha_{k''}}{\sigma_0}{\sigma_1} \Downarrow_P \Config{\Skip}{\sigma_0''}{\sigma_1''}$,
	\item $(e_{k''},\sigma_0''\cup\sigma_1'') \toExpr n$,
	and
	\item $\Config{\beta'}{\sigma_0'''}{\sigma_1'''} \Downarrow_P \Config{\Skip}{\sigma_0'}{\sigma_1'}$
\end{itemize}
and thus, 
$
\Config{\Assign{z'}{\FCall{\symb{f}_{k''}}{\Evec[k'']}};~\beta'}{\sigma_0}{\sigma_1} \Downarrow_P \Config{\Skip}{\sigma_0'}{\sigma_1'}
$ holds.
Therefore, the claim holds.
\qed
\end{proof}

%\subsection*{Proof of Theorem~\ref{thm:correctness}}
Correctness of $\Convert$ %Theorem~\ref{thm:correctness} 
can easily be proved by using Lemmas~\ref{lem:calc} and~\ref{lem:correctness}.
\begin{theorem}[Correctness of $\Convert$]
\label{thm:correctness}
Let %$i$ be a non-negative integer, 
$\cR = \Convert(P)$, % = \cR' \cup \{~\Rule{\var{u}}{\symb{return}(e)}~\}$ where $\Trans(\symb{f}(\Yvec), s, 1)= (u, \cR', j)$.
%$n'_1,\ldots,n'_m \in \Int$,
$n \in \Int$,
$s$ a normal form of $\cR$,
$i \in \{1,\ldots,k'\}$,
$\sigma_0,\sigma_0'$ assignments for $\Xvec$,
%$\sigma_0=\{x_1\mapsto n_1,\ldots,x_k\mapsto n_k\}$,
and
%$\sigma_1=\{y_1\mapsto n'_1,\ldots,y_m\mapsto n'_m\}$.
$\sigma_1,\sigma_1'$ assignments for $\Yvec[i]$.
Then,
$\Config{\alpha_i}{\sigma_0}{\sigma_1} \Downarrow_P \Config{\Skip}{\sigma_0'}{\sigma_1'}$ and $(e_i,\sigma_0'\cup\sigma_1') \toExpr n$
if and only if 
$
(\symb{env}(\Xvec,\symb{stack}(\symb{f}_i(\Yvec[i]),s)))(\sigma_0\cup\sigma_1)
%\mathrel{\to^*_{\cR}} (\symb{env}(\Xvec,\symb{stack}(\symb{return}(e,\bot)))(\sigma_0'\cup\sigma_1') 
\mathrel{\to^*_{\cR}} (\symb{env}(\Xvec,\symb{stack}(\symb{return}(n),s)))(\sigma_0'\cup\sigma_1')
$.
\end{theorem}
\begin{proof}
We first show the \textit{only-if} part.
Assume that 
$
\Config{\alpha_i}{\sigma_0}{\sigma_1} \Downarrow_P \Config{\Skip}{\sigma_0'}{\sigma_1'}
$
and
$
(e_i,\sigma_0'\cup\sigma_1') \toExpr n
$.
It follows from Lemma~\ref{lem:correctness} and $\Rule{C_i[\var{u}_i]}{C_i[\symb{return}(e_i)]} \in \cR$ (where $C_i[\,]$ is a context defined in Definition~\ref{def:converter}) that 
\[
\begin{array}{@{}l@{\>}l@{}}
\symb{env}((\Xvec)\sigma_0,\symb{stack}(\symb{f}_i((\Yvec[i])\sigma_1),s)) 
&
\mathrel{\to^*_{\cR}}
 \symb{env}((\Xvec)\sigma_0',\symb{stack}(u_i\sigma_1',s)) \\
&
\mathrel{\to_{\cR}} 
\symb{env}((\Xvec)\sigma_0',\symb{stack}(\symb{return}(e_i(\sigma_0'\cup\sigma_1')),s)).
\\
\end{array}
\]
It follows from Lemma~\ref{lem:calc} that $e_i(\sigma_0'\cup\sigma_1') \mathrel{\to^*_{\cR}} n$, and thus 
\[
\symb{env}((\Xvec)\sigma_0',\symb{stack}(\symb{return}(e_i(\sigma_0'\cup\sigma_1')),s))
\mathrel{\to^*_{\cR}} \symb{env}((\Xvec)\sigma_0',\symb{stack}(\symb{return}(n),s)).
\]
Therefore, the \textit{only-if} part holds.

Next, we show the \textit{if} part.
Assume that 
\[
\begin{array}{@{}l@{\>}l@{}}
\symb{env}((\Xvec)\sigma_0,\symb{stack}(\symb{f}_i((\Yvec[i])\sigma_1),s)) 
&
\mathrel{\to^*_{\cR}} 
\symb{env}((\Xvec)\sigma_0',\symb{stack}(u_i\sigma_1',s)) \\
&
\mathrel{\to^*_{\cR}} 
\symb{env}((\Xvec)\sigma_0',\symb{stack}(\symb{return}(e_i(\sigma_0'\cup\sigma_1')),s)) \\
&
\mathrel{\to^*_{\cR}} 
\symb{env}((\Xvec)\sigma_0',\symb{stack}(\symb{return}(n),s)).
\end{array}
\]
It follows from Lemmas~\ref{lem:correctness} and~\ref{lem:calc} that
$
\Config{\alpha_i}{\sigma_0}{\sigma_1} \Downarrow_P \Config{\Skip}{\sigma_0'}{\sigma_1'}
$
and
$
(e_i,\sigma_0'\cup\sigma_1') \toExpr n
$.
Therefore, the \textit{if} part holds.
\qed
\end{proof}
%The proof of Theorem~\ref{thm:correctness} can be seen in Appendix~\ref{sec:proof}.
%Theorem~\ref{thm:correctness} referred to the case where both the execution of $\FCall{\symb{f}}{n'_1,\ldots,n'_m}$ and the reduction from $(\symb{env}(\Xvec,\symb{stack}(\symb{f}(\Yvec),\bot)))(\sigma_0\cup\sigma_1)$ terminate.
%Though, 
Theorem~\ref{thm:correctness} implies that the execution of $\FCall{\symb{f}_i}{\Yvec[i]}$ with $\sigma_0,\sigma_1$ does not halt if and only if the reduction from $(\symb{env}(\Xvec,\symb{stack}(\symb{f}_i(\Yvec[i]),s)))(\sigma_0\cup\sigma_1)$ does not terminate.
This is because by the semantics, the execution of a program never halts unsuccessfully and either successfully halts or does not halt.
%Thus, the case where $\Config{\alpha}{\sigma_0}{\sigma_1} \Downarrow_P \Config{\Skip}{\sigma_0'}{\sigma_1'}$ does not hold for any assignments $\sigma_0'$ and $\sigma_1'$ coincides with the case where the execution does not hold.

%Intuitively, it holds that the execution of $\FCall{\symb{f}}{n'_1,\ldots,n'_m}$ halts if and only if the reduction from $(\symb{env}(\Xvec,\symb{stack}(\symb{f}(\Yvec),\bot)))(\sigma_0\cup\sigma_1)$ terminates.
%This is because the reduction of $\cR$ straightforwardly simulates the execution of $\symb{f}$ (see Figure~\ref{fig:sumrecstack-reduction}).
%However, the semantics of {\SIMP} programs we adopt is not adequate to capture non-terminating execution of programs.

%%%%%%%%%%%%%%%%%%%%%%%%%%%%%%%%%%%%%%%%%%%%%%%%%%%%%%%%%%%%%%%%%%%%%%%%%
\section{Conclusion}
\label{sec:conclusion}
%%%%%%%%%%%%%%%%%%%%%%%%%%%%%%%%%%%%%%%%%%%%%%%%%%%%%%%%%%%%%%%%%%%%%%%%%

In this paper, we proposed a new transformation of imperative programs with function calls and global variables into LCTRSs, and proved correctness of the transformation.
%We will formalize the transformation based on the approach, proving correctness of the transformation.
A direction of future work is to apply the new transformation to a sequential program and its parallelized version in order to prove their equivalence.
To simplify the discussion, we considered a program executed as a single process, i.e., executed \emph{sequentially}, and the introduced symbol $\symb{env}$ has an argument that is used for the single process (see $\cRsumrecstack$ again). 
To adapt to \emph{parallel} execution where the number of executed processes is fixed, it suffices to add arguments for all executed processes into the symbol $\symb{env}$.
We will formalize this idea and prove the correctness of the transformation for parallel execution.

\paragraph{Acknowledgements}
We gratefully acknowledge the anonymous reviewers for their useful comments and suggestions to improve the paper.

%%%%%%%%%%%%%%%%%%%%%%%%%%%%%%%%%%%%%%%%%%%%%%%%%%%%%%%%%%%%%%%%%%%%%%%%%
%\bibliographystyle{eptcs}
%\bibliography{biblio}

%%%%%%%%%%%%%%%%%%%%%%%%%%%%%%%%%%%%%%%%%%%%%%%%%%%%%%%%%%%%%%%%%%%%%%%%%

%%%%%%%%%%%%%%%%%%%%%%%%%%%%%%%%%%%%%%%%%%%%%%%%%%%%%%%%%%%%%%%%%%%%%%%%%
\end{document}